\documentclass[11pt]{article} 
\usepackage{preamble}

\newcommand{\Lloc}{\lambda_{\mathrm{loc}}}
\newcommand{\rhoin}{\rho_{\mathrm{in}}}
\newcommand{\rhoout}{\rho_{\mathrm{out}}}
\newcommand{\bits}{\mathbb{F}_2}
\newcommand{\Ddeg}{\mathfrak{d}_{\mathrm{dis}}}
\newcommand{\Kdeg}{K_{\mathrm{deg}}}
\newcommand{\Kden}{K_{\mathrm{den}}}
\newcommand{\Krec}{K_{\mathrm{rec}}}
\newcommand{\Ttot}{T_{\mathrm{tot}}}

\title{Near-Optimal Learning of Local Lindbladians}
\author[1]{Itai Arad\thanks{\texttt{arad.itai@fastmail.com}}}
\author[1]{Zhili Chen\thanks{\texttt{chen.zhili@u.nus.edu}}}
\author[1]{Naixu Guo\thanks{\texttt{naixug@u.nus.edu}}}
\author[1,2]{Patrick Rebentrost\thanks{\texttt{cqtfpr@nus.edu.sg}}}
\author[1]{Zhan Yu\thanks{\texttt{yu.zhan@u.nus.edu}}}
\affil[1]{\normalsize Centre for Quantum Technologies, National University of Singapore, 117543, Singapore}
\affil[2]{School of Computing, National University of Singapore, 117417, Singapore}

\begin{document}
\pagenumbering{roman}
\maketitle
\thispagestyle{empty}

\begin{abstract}
    We study the problem of learning local Lindbladians from black-box access to the physical evolution, where the goal is to estimate all Hamiltonian and dissipative coefficients. For local Lindbladians with low-intersection dissipators and local dynamical strength at most $\Lambda$, we design an algorithm that learns the coefficients to target accuracy $\eps$, whose experiment count and total evolution time scale as $\widetilde{O}(\Lambda^2/\eps^2)$ and $\widetilde{O}(\Lambda/\eps^2)$, respectively, with only logarithmic dependence on the number of qubits. The algorithm is built directly from finite-time channel probes. It runs the unknown evolution for short times, estimates the corresponding Pauli transfer matrices from classical shadows, and converts these estimates into Lindbladian coefficients by stable local Fourier inversions.
    
    The algorithm is non-adaptive, uses no ancillas, and requires only random product states as inputs followed by random Pauli measurements. The method does not require knowing the structure of the Lindbladian in advance. We prove matching lower bounds that establish the near-optimality of the algorithm in both resources. By constructing a family of single-qubit dephasing Lindbladians, we show that any algorithm, even an adaptive one with arbitrary ancillas and measurements, requires $\Omega(\Lambda^2/\eps^2)$ experiments and $\Omega(\Lambda/\eps^2)$ total evolution time. In particular, the lower bounds imply that the Heisenberg-limited scaling achievable for Hamiltonian learning is information-theoretically impossible once dissipative coefficients must be estimated.
\end{abstract}

\clearpage
\pagestyle{empty}
\tableofcontents
\clearpage

\pagestyle{plain}
\pagenumbering{arabic}
\section{Introduction}
Learning the dynamical laws governing a quantum system is a central task across quantum science, from simulation and computation \cite{bernien2017probing, lloyd1996universal,shulman2014suppressing} to metrology and many-body physics \cite{leibfried2004toward, Wiebe2014hamiltonian,wang2017experimental}. From an algorithmic point of view, the key challenge is to identify the generator of the dynamics using as few experiments as possible, while avoiding a description whose size is exponential in the number of qubits.

In the ideal regime of a closed quantum system, the evolution is unitary and is completely characterized by a Hamiltonian $H$. Expanding $H = \sum_a h_a P_a$ in the Pauli basis, learning the dynamics amounts to estimating the interaction coefficients $h_a$. This Hamiltonian learning problem has been studied extensively and has led to algorithms with strong guarantees on sample complexity and evolution time under various assumptions on locality, temperature, measurements, and available control~\cite{anshu2020sample, zubida2021optimal, haah2022optimal, huang2023learning, bakshi2024learning, hu2025ansatz}.

However, physical systems are rarely perfectly isolated from their surroundings. Interaction with the environment causes decoherence, relaxation, dephasing, and other irreversible effects. These phenomena cannot be represented by a Hamiltonian alone, since Hamiltonian evolution is reversible and preserves purity, whereas open quantum system evolution is generally noisy and dissipative. Thus, a learning framework that only targets Hamiltonians does not capture the dynamics of realistic quantum devices.

For Markovian open quantum systems, the standard model is a \emph{Lindbladian} generator. The Gorini--Kossakowski--Sudarshan--Lindblad (GKSL) theorem states that every finite-dimensional, time-homogeneous, completely positive and trace-preserving quantum Markov semigroup has a generator of Lindblad form~\cite{gorini1976completely,lindblad1976generators}:
\begin{equation*}
\cL(\rho) = \underbrace{\sum_{P_a \neq I}-\iu h_a [P_a,\rho]}_{\text{Hamiltonian part}}\quad + \quad \underbrace{\sum_{P_b,P_c \neq I} \gamma_{bc} \ab\Big(P_b\rho P_c-\frac{1}{2}\{P_cP_b,\rho\})}_{\text{Dissipative part}}.
\end{equation*}
A Lindbladian contains both a coherent Hamiltonian part and a dissipative part, and therefore gives a unified description of reversible interactions and irreversible environmental effects. This makes Lindbladians a natural and widely used model for open-system many-body dynamics, noisy quantum hardware, engineered dissipation and dissipative state preparation, and open-system quantum simulators~\cite{barreiro2011anopen, harper2020efficient, chen2025efficientquantum, ding2025end, rouze2025efficient, shang2025fast, shang2025designing}. Learning a Lindbladian is therefore the natural open-system analogue of Hamiltonian learning: rather than only estimating the Hamiltonian coefficients $h_a$, one must also estimate the dissipative coefficients $\gamma_{bc}$ that specify the noise.

A common and physically relevant setting is that of \emph{local} Lindbladians, which govern the dynamics of most realistic open quantum systems. In a local Lindbladian, each elementary Hamiltonian or dissipative term acts on at most a constant number of qubits. This captures the usual few-body interactions in quantum many-body systems and the dominant local noise processes in quantum devices, such as single-qubit relaxation and dephasing together with a limited number of correlated few-qubit errors. In particular, the dissipative part is often \emph{low-intersection}, i.e., each qubit participates in only a bounded number of dissipative terms, since in real devices each qubit has a small number of dominant relaxation, dephasing, leakage, or correlated-error channels, constrained by the hardware connectivity and by the way gates are implemented~\cite{acharya2025quantum, harper2020efficient}.
This leads to the first fundamental question that we focus on in this work:
\begin{question}\label{q:upper}
Can one efficiently learn local Lindbladians from black-box access to their time evolution?
\end{question}
Here the efficiency is measured by the physical resources used by the learner. Following the Hamiltonian learning literature, we consider two main resources: the \emph{experiment count}, which is the number of invocations of the unknown evolution, and the \emph{total evolution time}, which is the accumulated physical time that the unknown dynamics is run for.
Beyond designing efficient algorithms, we seek to determine the physical resources fundamentally required for Lindbladian learning, regardless of the learning strategy. The two resources require separate lower bounds. Since the duration of each experiment can vary, a lower bound on total evolution time alone does not determine how many experiments are necessary. This leads to our second question:
\begin{question}\label{q:lower}
What are the lower bounds on the experiment count and total evolution time required for Lindbladian learning?
\end{question}

We answer \cref{q:upper} affirmatively by designing an algorithm that learns coefficients of local Lindbladians without prior knowledge of the structure, using only non-adaptive, ancilla-free experiments.
We then resolve \cref{q:lower} by proving information-theoretic lower bounds on the experiment count and total evolution time, valid for arbitrary adaptive quantum algorithms. For learning local Lindbladians with low-intersection dissipators, these bounds match our algorithm's dependence on dynamical strength and target precision up to logarithmic factors.
The next subsection specifies the access model and the learning problem, and states our main results.

\subsection{Results}
Our main result is a near-optimal algorithm for the local Lindbladian learning problem, that is, to estimate all Hamiltonian and dissipative coefficients of an unknown $k$-local Lindbladian $\cL$ from black-box access to its time evolution $\{e^{t\cL}\}_{t\geq 0}$, together with matching information-theoretic lower bounds. Before stating the results, we describe the access model and the learning problem. Formal definitions are given in \cref{sec:background}.

\paragraph{Access model and resource measures.}
In each experiment, the learner selects an evolution time $t\geq0$, prepares an input state, applies the unknown channel $e^{t\cL}$ once, and measures the output. The \emph{experiment count} is the number of such channel applications, and an experiment with evolution time $t_j$ contributes $t_j$ to the \emph{total evolution time} $\Ttot=\sum_j t_j$. We also track the \emph{time resolution}, which is the shortest evolution time used across all experiments. Protocols whose time resolution decreases polynomially with the target accuracy are difficult to realize in practice. Our lower bounds hold in the strongest form of this access model, where the learner may attach arbitrarily many ancillas, apply arbitrary quantum operations between experiments, and choose every evolution time and measurement adaptively. In contrast, our algorithm only requires the weakest form of this model. It is non-adaptive and ancilla-free, with product-state inputs and single-qubit Pauli measurements.

\paragraph{Local Lindbladian learning.}
Before formally stating the learning problem, let us first define the class of Lindbladians under consideration; formal definitions are given in \cref{sec:background}. We consider a $k$-local Lindbladian for a constant $k$: every Hamiltonian term $P_a$ and every dissipative Pauli pair $(P_b,P_c)$ acts nontrivially on at most $k$ qubits, i.e., the weight of Pauli terms satisfies $\wt(P_a) \leq k$ and $\wt(P_b, P_c) \leq k$. We emphasize that no geometric locality is required and the $k$-qubit supports may be scattered arbitrarily. We only assume that the total strength of all interactions and noise processes touching any single qubit is bounded. Decomposing $\cL=\sum_S\cL_S$, where $\cL_S$ collects the terms supported exactly on the region $S$, the \emph{local dynamical strength} satisfies
\begin{equation}\label{eq:local_strength_intro}
  \Lloc\coloneq\max_{i\in[n]}\sum_{S\ni i}\norm[\big]{\cL_S^\dagger}_{\infty\to\infty} \leq \Lambda.
\end{equation}
The parameter $\Lloc$ is the open-system analogue of the one-spin energy used in Hamiltonian learning~\cite{bakshi2024structure}. Our algorithm will also depend on the \emph{dissipative site degree} $\Ddeg$, the maximum number of dissipative terms in which any single qubit participates:
\begin{equation}\label{eq:site_degree_intro}
  \Ddeg\coloneq\max_{i\in[n]}\ \abs[\big]{\cbra{(b,c):\gamma_{bc}\neq0,\ i\in\supp(P_b) \cup \supp(P_c)}}.
\end{equation}
When $\Ddeg = O(1)$, we say the Lindbladian has \emph{low-intersection dissipators}. This mirrors realistic devices, where each qubit suffers from a few dominant noise channels~\cite{acharya2025quantum}. 

The learner is given only the parameters $(k, \Ddeg, \Lambda)$, and has no prior knowledge of which Hamiltonian or dissipative terms are present in $\cL$. The learning task is formally defined as follows.

\begin{problem}[Lindbladian learning]\label{prob:Lindbladian_learning}
Fix $k=O(1)$, let $\cL$ be an unknown $k$-local Lindbladian on $n$ qubits with dissipative site degree $\Ddeg$ and local dynamical strength at most $\Lambda$ (\cref{eq:local_strength_intro,eq:site_degree_intro}), whose Hamiltonian coefficients are $\{h_a\}$ and dissipative coefficients are $\{\gamma_{bc}\}$. Given $\eps,\delta>0$, along with adaptive black-box access to the semigroup $\{e^{t\cL}\}_{t\geq 0}$, find estimates $\{\hat h_a\}$ and $\{\hat\gamma_{bc}\}$ such that
\[
  \max_{a}{\abs{\hat{h}_a-h_a}}\leq \eps,
  \qquad
  \max_{(b,c)}{\abs{\hat\gamma_{bc}-\gamma_{bc}}} \leq \eps
\]
with probability at least $1-\delta$.
\end{problem}

Our first main result solves \cref{prob:Lindbladian_learning} with near-optimal resources, answering \cref{q:upper}. We summarize its guarantees informally as follows.

\begin{theorem}[Upper bound for local Lindbladian learning; informal version of \cref{thm:learning_upper}]\label{thm:informal_upper}
Fix $k=O(1)$, and let $\cL$ be an unknown $k$-local Lindbladian on $n$ qubits with bounded dissipative site degree $\Ddeg=O(1)$ and local dynamical strength at most $\Lambda$. For any $\eps,\delta\in(0,1)$, there is a learning algorithm with the following guarantees:
\begin{enumerate}
    \item \emph{(Accuracy)} the outputs satisfy $\max_a{\abs{\hat h_a-h_a}}\leq\eps$ and $\max_{(b,c)}{\abs{\hat\gamma_{bc}-\gamma_{bc}}}\leq\eps$ with probability at least $1-\delta$;
    \item \emph{(Experiment count)} it performs $\widetilde O\ab\big(\frac{\Lambda^2}{\eps^2}\log\frac{n}{\delta})$ experiments;
    \item \emph{(Total evolution time)} its total evolution time is $\Ttot=\widetilde O\ab\big(\frac{\Lambda}{\eps^2}\log\frac{n}{\delta})$;
    \item \emph{(Time resolution)} every experiment evolves the system for time $t=\widetilde\Theta(1/\Lambda)$, and only $O(\log(\Lambda/\eps))$ distinct evolution times are ever used;
    \item \emph{(Simple experimental primitives)} the experiments are non-adaptive and ancilla-free: one prepares a uniformly random Pauli product eigenstate, evolves, and measures each qubit in a random Pauli basis;
    \item \emph{(Classical overhead)} the classical post-processing time is $\widetilde O\ab\big(n^k\frac{\Lambda^2}{\eps^2}\log\frac{n}{\delta})$.
\end{enumerate}
\end{theorem}

\begin{remark}[Generic local dissipators]
The low-intersection assumption is applied only in the coefficient-recovery stage. In general, the bounds on the experiment count, total evolution time, and classical running time contain $\Krec^2$, where the recovery factor $\Krec=\min\{\Kdeg,\Kden\}$ satisfies $\Kdeg=2\sum_{j=0}^{k-1}(2\Ddeg)^j$, which is $O(1)$ whenever $\Ddeg=O(1)$, and $\Kden=O(n^{k-1})$ unconditionally. Therefore, the cost of the algorithm remains polynomial in all parameters even without any low-intersection assumption.
\end{remark}

We complement the algorithm with lower bounds that answer \cref{q:lower}, showing that both resource scalings are optimal up to logarithmic factors, and that they remain so in the strongest access model. The formal statements are \cref{thm:lower_bound_measurement} (experiment count), \cref{thm:lower_bound_time} (total evolution time), and \cref{cor:expected_time_lower} (expected evolution time of fully adaptive strategies).

\begin{theorem}[Lower bounds for local Lindbladian learning; informal version of \cref{thm:lower_bound_measurement,thm:lower_bound_time,cor:expected_time_lower}]\label{thm:informal_lower}
Let $\Lambda>0$ and $\eps\in(0,\Lambda/16]$. Any algorithm that estimates the coefficients of unknown $k$-local Lindbladians with local dynamical strength at most $\Lambda$ up to $\ell^\infty$ error $\eps$ with probability at least $2/3$ must perform $\Omega(\Lambda^2/\eps^2)$ experiments, and its (expected) total evolution time must be $\Omega(\Lambda/\eps^2)$. This holds even for adaptive algorithms that use arbitrary ancillas, entangling operations, entangled measurements, and adaptively chosen evolution times.
\end{theorem}

In particular, \cref{thm:informal_lower} shows that once dissipative coefficients must be estimated, the total evolution time is pinned to the standard quantum limit $1/\eps^2$, so the Heisenberg-limited scaling of Hamiltonian learning is information-theoretically impossible for Lindbladian learning.

\subsection{Related work}

\paragraph{Hamiltonian learning.}
Characterizing unknown Hamiltonians has been studied extensively, with early works exploring machine-learning and quantum-simulation-based approaches for practical characterization \cite{dasilva2011practical, henstchel2010machine, granade2012robust, Wiebe2014hamiltonian}, and later formulations estimating the coefficients of Hamiltonians under different forms of access to the system, including real-time dynamics, thermal states, and eigenstates \cite{gu2024practical, bairey2019learning, zubida2021optimal, qi2019determining, evans2019scalable}.
In the dynamical setting, $n$-qubit $k$-local Hamiltonians with low-intersection interaction terms have been shown to be learnable with Heisenberg-limited scaling $O(1/\eps)$ \cite{huang2023learning}, while subsequent work established that quantum control is necessary to attain this scaling in general \cite{dutkiewicz2024advantage}.
Further developments have relaxed, or even removed, the assumption of local interactions \cite{zhao2025learning, hu2025ansatz, Sinha2025improved, ma2024learning}, explored learning from long-time evolution rather than many short-time experiments \cite{shin2026heisenberg, pradenne2026learning}, proved a lower bound of $(n/k)^{\Omega(k)}$ on the total evolution time~\cite{chen2026lower}, and learned the underlying Hamiltonian by estimating the short-time Pauli transfer matrix~\cite{zubida2021optimal,franca2024efficient,caro2024learning}.
Another line of work considers Hamiltonian learning from Gibbs states $e^{-\beta H}$: sample-efficient learning was first established in \cite{anshu2020sample}, and computationally efficient algorithms with optimal sample complexity were later obtained in the high-temperature regime \cite{haah2022optimal}, which also exhibits a ``sudden death'' of thermal entanglement \cite{bakshi2024high}, and were extended to the constant-temperature regime $\beta=O(1)$ \cite{bakshi2024learning}.
Other variants, including structure learning \cite{bakshi2024structure}, robust learning \cite{yu2023robust}, certification \cite{gao2026quantum, bluhm2026certifying}, and property testing \cite{bluhm2026hamiltonian, arunachalam2025testing}, have also been considered, and experimental demonstrations have been reported \cite{wang2017experimental, gentile2021learning, hangleiter2024robustly, guo2025hamiltonian}.

\paragraph{Lindbladian learning.}
Early attempts to reconstruct Markovian open-system dynamics from measured data date back over two decades \cite{buzek1998reconstruction, boulant2003robust}. Recent work includes learning from steady states~\cite{bairey2020learning}, heuristic and numerical procedures for fitting noise models to tomographic data~\cite{onorati2023fitting, vandenberg2024techniques, liu2025robust}, experimental demonstrations~\cite{birke2026demonstrating}, detection of dissipation with Heisenberg-limited scaling~\cite{cai2026optimal}, and non-Markovian noise learning~\cite{montana-lopez2025efficiently}. The works most relevant to ours learn Lindbladian coefficients from access to the time evolution~\cite{franca2024efficient, franca2025learning, ivashkov2026ansatz}. These protocols share the broad template of estimating short-time derivatives and then inverting a linear system, but differ in their structural assumptions and in the conditioning of the inverse problem.

Fran\c{c}a et al.~\cite{franca2024efficient} showed that geometrically local Lindbladians obeying a Lieb--Robinson bound can be learned at the standard quantum limit. Combining polynomial interpolation of time derivatives with a shadow process tomography scheme, their protocol estimates the coefficients to accuracy $\eps$ using $\widetilde O(1/\eps^2)$ samples. This interpolation and Lieb--Robinson framework was subsequently extended to learning and certification of time-dependent local dynamics~\cite{franca2025learning}.
Their learning protocols are limited to simple single-qubit dissipators, the conditioning of recovering the general two-sided dissipative coefficients is not analyzed, and the degree of the interaction graph is assumed to be a constant. In contrast, our result applies to arbitrary constant-local dissipators, and makes no assumption on the geometrical locality.

Ivashkov et~al.~\cite{ivashkov2026ansatz} initiated \emph{ansatz-free} Lindbladian learning, targeting $M$-sparse Lindbladians without any locality assumption. This is a broader and intrinsically harder problem than the local setting considered here, but their end-to-end Lindbladian learning protocol performs $\widetilde O(M^4 \nu^2/\eps^4)$ experiments, where $\nu$ is an instance-dependent conditioning factor. The factor $\nu$ cannot be removed in their algorithm even for local Lindbladians, and it may scale exponentially in the underlying dimension in the worst case. It has been shown in~\cite{montana-lopez2025efficiently} that the conditioning factor can be bounded for the $k$-local and low-intersection cases, but this requires the knowledge of the supports. Our learning algorithm avoids this bottleneck of conditioning by decoupling the inverse problem into local Walsh--Hadamard transforms with condition number one, followed by a controlled de-aliasing step.

Brady et al.~\cite{brady2026precision} derive quantum Fisher information bounds for multiparameter Markovian-noise metrology, allowing arbitrary quantum control and noiseless ancillas. Specialized to single-qubit dephasing at rate $\gamma=\Theta(\Lambda)$, their bounds imply the lower bound $\Omega(\Lambda/\eps^2)$ on total evolution time through the quantum Cram\'er--Rao bound for only unbiased estimators with mean-square error at most $\eps^2$. Also, the QFI bound alone does not imply the lower bound on the experiment count, since it grows with the duration of a single experiment.

\paragraph{Concurrent work.}
While completing this manuscript, we became aware of several independent and concurrent works on Lindbladian learning~\cite{romanov2026learning, sinha2026efficient, mobus2026robust, lewis2026learning}.

Romanov et~al.~\cite{romanov2026learning} give the first standard-quantum-limited algorithm for the ansatz-free setting of~\cite{ivashkov2026ansatz}, learning arbitrary $M$-sparse Lindbladians with total evolution time $\widetilde O(M/\eps^2)$, and, under an additional regularity condition, learning part of the Hamiltonian at the Heisenberg limit. They remove the conditioning bottleneck of~\cite{ivashkov2026ansatz} through a recursive quantum error correction construction, which requires ancillary qubits and Clifford control interleaved with the evolution. In the local regime, our protocol attains the standard quantum limit with no ancillas or mid-evolution control, and with only logarithmic dependence on the system size. Sinha~\cite{sinha2026efficient} also revisits the ansatz-free setting, with an emphasis on robustness to state-preparation-and-measurement errors.

M\"obus et~al.~\cite{mobus2026robust} consider the same $k$-local problem through a similar route of short-time PTM estimation followed by local inversion, and additionally project the estimate onto a valid GKSL generator via a semidefinite program. Treating the local dynamical strength as a constant, their general support-free algorithm uses $\widetilde O(n^{2k}/\eps^2)$ samples, while this improves to $\widetilde O(\log(n)/\eps^2)$ samples when a candidate support with bounded intersection degree is supplied. They additionally establish a guarded structure-learning guarantee without prior support knowledge under the stable-tail condition.  This condition controls the amplification of unresolved diagonal tails and is not implied by bounded degree alone. Their guarded structure-learning procedure uses $\widetilde O(\Ddeg^{2k}\log(n)/\eps^2)$ samples. Consequently, when $\Ddeg=O(1)$, it achieves $\widetilde O(\log(n)/\eps^2)$ sample complexity, provided the stable-tail condition holds. This achieves a similar complexity to our result. In comparison, our algorithm does not need the stable-tail condition, and our lower bounds hold against adaptive protocols with arbitrary ancillas, while theirs are restricted to product inputs and Pauli measurements.

Lewis, Tang, and Wright~\cite{lewis2026learning} learn constant-local Lindbladians with an algorithm in a setting similar to ours. In terms of the approximate interaction-graph degree $d$ and the single-site energy $g$ (the analogue of our parameter $\Lambda$), the total evolution time is $\widetilde{O}(gd^2/\eps^2)$. We emphasize a key difference in the degree parameters: the parameter $d$ in~\cite{lewis2026learning} is the degree of the \emph{entire} interaction graph, counting Hamiltonian and dissipative terms together, whereas our $\Ddeg$ constrains only the dissipative part, and we impose no degree bound on the Hamiltonian part. Consequently, for systems whose Hamiltonian interaction graph is dense but whose noise consists of a few dominant channels per qubit, the degree $d$ can grow polynomially with the system size while $\Ddeg$ stays constant. When the entire interaction graph has bounded degree, their algorithm attains total evolution time $\widetilde O(g/\eps^2)$ and saturates our lower bound, yielding the same performance as our algorithm. In terms of the dependence on the respective degree parameters, the algorithm of~\cite{lewis2026learning} scales polynomially as $d^2$, improving on the $\Ddeg^{2k-2}$ factor of our algorithm, while our classical post-processing is polynomial and theirs is quasi-polynomial in general.

\subsection{Technical overview}

Our algorithm accesses the unknown dynamics only through non-adaptive randomized experiments at a logarithmic number of short evolution times, followed by classical post-processing, in three stages,
\[
  \{e^{t\cL}\}_{t\geq0}
  \ \xrightarrow{\ \text{probe}\ }\
  \cbra[\big]{\bar R_{vw}(t_j)}
  \ \xrightarrow{\ \text{differentiate}\ }\
  \cbra[\big]{\hat L_{vw}}
  \ \xrightarrow{\ \text{invert \& de-aliasing}\ }\
  \cbra[\big]{\hat h_a}\cup\cbra[\big]{\hat\gamma_{bc}}.
\]
We describe the three stages in turn and then explain the ideas behind the matching lower bounds. Compared with Hamiltonian learning, the two main technical differences are (i) derivative bounds that depend on the local rather than the global dynamical strength, and (ii) a coefficient-recovery step, which we reduce to well-conditioned local Fourier transforms followed by a de-aliasing procedure.

\paragraph{Estimating the generator from finite-time probes.}
We cannot access the generator $\cL$ directly but only observe the finite-time dynamics $e^{t\cL}$. The two are connected by the continuous-time Pauli transfer matrix (PTM),
\[
  R_{vw}(t)\coloneq\frac{1}{2^n}\tr\ab\big(P_v\,e^{t\cL}(P_w)), \qquad
  L_{vw}\coloneq R'_{vw}(0)=\frac{1}{2^n}\tr\ab\big(P_v\,\cL(P_w)),
\]
for Pauli operators $P_v,P_w$. Since the coefficients will be recovered region by region, we only need the $O(n^k)$ entries indexed by $\cT_k=\{(v,w):\wt(v,w)\leq k\}$. To estimate the finite-time entries, we adapt an ancilla-free shadow process tomography protocol~\cite{franca2024efficient,levy2024classical,kunjummen2023shadow} (\cref{sec:ptm_learning}). Each experiment prepares a uniformly random product state of single-qubit Pauli eigenstates, applies $e^{t\cL}$, and measures every qubit in an independent uniformly random Pauli basis. Each measurement record yields, for every target pair $(v,w)\in\cT_k$, an unbiased single-shot estimator $\hat R_{vw}(t)$ with second moment at most $3^{2k}$ (\cref{lem:unbiasedness_of_R,lem:variance}). One batch of records is therefore reused for all $O(n^k)$ target entries through different classical post-processing, without requiring additional experiments. The difficulties specific to open systems appear later, when the generator entries are converted into physical coefficients.

The generator entry $L_{vw}=R'_{vw}(0)$ is a derivative at the boundary $t=0$, which must be estimated from the noisy finite-time estimates (\cref{sec:chebyshev}). The naive finite difference $[R_{vw}(\tau)-R_{vw}(0)]/\tau$ suffers from a bias--variance trade-off: its bias is $O(\Lambda^2\tau)$, which forces $\tau=O(\eps/\Lambda^2)$, and dividing the statistical error by such a small $\tau$ increases the sample complexity to $O(1/\eps^4)$~\cite{bairey2020learning,zubida2021optimal}. Instead, following the interpolation technique of recent generator-learning algorithms~\cite{caro2024learning,franca2024efficient,gu2024practical,ivashkov2026ansatz,he2026optimal}, we place $q+1$ Chebyshev--Lobatto nodes $t_j=\frac{T}{2}(1-\cos\frac{j\pi}{q})$ on a short interval $[0,T]$ and differentiate the degree-$q$ interpolating polynomial at $t=0$. This gives an explicit endpoint differentiation rule $\hat L_{vw}=\sum_{j=0}^{q}\ell_j'(0)\,\bar R_{vw}(t_j)$, a linear combination of the empirical PTM estimates $\bar R_{vw}(t_j)$ with Lagrange weights $\ell_j'(0)$. The value at $t_0=0$ is known exactly since $R_{vw}(0)=\mathbf{1}[v=w]$, so experiments are only run at the $q$ positive nodes. The estimation error consists of the interpolation bias, which is controlled by the $(q+1)$-st time derivative of $R_{vw}$, and the statistical error, which is amplified by the weight sum $\sum_{j=1}^q\abs{\ell_j'(0)}=O(q^2/T)$. The interval length $T$ must therefore balance the growth of the derivatives against that of the weights.

\paragraph{Bounding the derivatives by the local dynamical strength.}
For a generic observable, the best available bound on the derivatives of $R_{vw}$ is $\abs{R^{(r)}_{vw}(t)}\leq\norm{\cL^\dagger}^r_{\infty\to\infty}$ in terms of the \emph{global} dynamical strength, which grows as $\Theta(n)$ for an extensive system and would make both the interval length $T$ and the final complexity scale polynomially with the system size. Our first key observation is that locality removes this dependence: for every $(v,w)\in\cT_k$, $r\in\N$, and $t\geq0$, it holds that
\[\abs[\big]{R^{(r)}_{vw}(t)}\leq r!\,(k\Lloc)^r.\]

The proof works in the Heisenberg picture, where $R^{(r)}_{vw}(t)=2^{-n}\tr\ab\big((\cL^\dagger)^r(P_v)\,e^{t\cL}(P_w))$: expanding $(\cL^\dagger)^r$ into $r$ layers of local terms acting on the $k$-local operator $P_v$, a summand is nonzero only if each layer intersects the support accumulated by the previous ones, which has size at most $\ell k$ before the $\ell$-th layer, so the $\ell$-th layer contributes a factor of at most $\ell k\Lloc$ and the $r$ layers together give $r!\,(k\Lloc)^r$. The factorial growth implies that $t\mapsto R_{vw}(t)$ is analytic with a radius of order $1/(k\Lloc)$. Fixing the interval length $T=1/(2k\Lambda)$, the interpolation error decays geometrically as $\abs[\big]{L_{vw}-p_q'(0)}\leq(k\Lambda)^{q+1}T^{q}=k\Lambda\cdot2^{-q}$, so $q=\lceil\log(2k\Lambda/\eps)\rceil$ nodes suffice to make the bias at most $\eps/2$ (\cref{lem:truncation}). This choice of parameters determines all the complexity bounds in \cref{thm:informal_upper}. The evolution times satisfy $t_j\in[\Theta(T/q^2),\,T]$, so the time resolution is $\widetilde\Theta(1/\Lambda)$, independent of $\eps$ up to logarithmic factors. The differentiation weights sum to $O(k\Lambda q^2)$, so the Bernstein bound gives $\widetilde O\ab\big(3^{2k}k^2\frac{\Lambda^2}{\eps^2}(k\log n+\log\frac{1}{\delta}))$ experiments. Since each experiment evolves for time at most $T=O(1/(k\Lambda))$, the total evolution time is smaller than the experiment count by a factor of $\Lambda$, which gives $\widetilde O(\Lambda/\eps^2)$ (\cref{thm:ptm_generator_estimation}).

\paragraph{Recovering the coefficients by local Fourier inversion.}
It remains to convert the estimated generator entries $\{\hat L_{vw}\}$ into estimates of the physical coefficients, and this step is where Lindbladian learning differs most from Hamiltonian learning. For closed dynamics, each short-time expectation value directly encodes a single Hamiltonian coefficient, whereas each Lindbladian entry $L_{vw}$ is a linear combination of many dissipative coefficients $\gamma_{bc}$. Prior approaches therefore invert an instance-dependent linear system whose condition number is either not analyzed or can be exponentially large in the worst case~\cite{franca2024efficient,ivashkov2026ansatz,montana-lopez2025efficiently}. We avoid this issue by choosing bases in which the linear map becomes an orthogonal transform (\cref{sec:coeff_recovery}).

Write $\cL$ in the left--right Pauli ($\chi$-matrix) basis, $\cL(A)=\sum_{b,c}\chi_{bc}P_b A P_c$. The $\chi$-matrix entries determine the physical coefficients directly (\cref{lem:chi_to_params}) as $h_a=\frac{\iu}{2}\ab\big(\chi_{a0}-\chi_{0a})$ and $\gamma_{bc}=\chi_{bc}$. Substituting the $\chi$-expansion into the PTM generator gives $L_{vw}=\sum_{b,c}\chi_{bc}\cdot2^{-n}\tr(P_v P_b P_w P_c)$, and the trace vanishes unless $v\oplus b\oplus w\oplus c=0$. Consequently, the linear map from $\{L_{vw}\}$ to $\{\chi_{bc}\}$ decouples according to the \emph{Pauli shift} $u\coloneq v\oplus w=b\oplus c$: only entries with the same shift mix. For a fixed shift $u$, absorbing explicitly computable Pauli phases into $F_u(w)\propto L_{w\oplus u,w}$ and $G_u(b)\propto\chi_{b, b\oplus u}$, the relation collapses to
\[
  F_u(w)=\sum_{b}G_u(b)\,(-1)^{\sinprod{b}{w}},
\]
which is a Walsh--Hadamard (discrete Fourier) transform over Pauli labels, where the symplectic inner product plays the role of the Boolean character. Restricting to a region $S$ of at most $k$ qubits turns this into the matrix identity $F^S_u=H_SG^S_u$ with $H_S^\top H_S=4^{\abs{S}}I$. The inversion \cref{eq:fwht_inverse} is thus a scaled orthogonal transform with condition number exactly $1$, so an entrywise error of $\eta$ on the PTM generator estimates leads to an error of at most $\eta$ on each inverted value, without any amplification. Moreover, each block has dimension $4^{\abs{S}}\leq4^k=O(1)$ and can be inverted by a fast Walsh--Hadamard transform at small classical cost (\cref{lem:recovery_time}).

\paragraph{De-aliasing without knowing the support.}
One issue remains for the local inversion: the Fourier transform on a region $S$ cannot distinguish a Pauli pair from its \emph{aliases}, namely the pairs $(b',b'\oplus u)$ with strictly larger support whose labels agree with $b$ on $S$ and coincide with each other outside $S$. The local inversion actually returns the aggregated sum
\[
  \chi^S_{b,b\oplus u} = \chi_{b,b\oplus u} + \sum_{\substack{b'\neq b:\ b'|_S=b\\ \gamma_{b',b'\oplus u}\neq0}}\chi_{b',b'\oplus u}.
\]
This aliasing only occurs for dissipative dynamics: for closed systems all outside extensions vanish, so the issue does not arise in short-time Hamiltonian learning. Without knowing the support of $\cL$, we remove the aliases by a thresholded peeling procedure (\cref{alg:coeff_recovery_threshold}) that processes candidates from union support size $k$ down to $1$: at size $s$, it subtracts the previously retained larger aliases from the aggregated sum and thresholds the residual at $\tau_s=\eta+\Ddeg\,E_{s+1}$, where $\eta$ is the accuracy of the PTM generator estimates and $E_s=2\tau_s$ is the reconstruction error at size $s$. The threshold controls the number of subtracted terms: a candidate has $\Theta(n^{k-s})$ potential aliases, but every nonzero alias acts on each qubit of $S$ in the same way as $(b,b\oplus u)$, so the dissipative site degree bounds their number by $\Ddeg$. A descending induction shows that the reconstruction error satisfies $E_s=2(\eta+\Ddeg E_{s+1})$, which gives $E_1=2\eta\sum_{j=0}^{k-1}(2\Ddeg)^j=\Kdeg\,\eta$, so the error only increases by a constant factor when $k,\Ddeg=O(1)$ (\cref{lem:recovery_error}). The same bound applies to the Hamiltonian coefficients, since their boundary entries $\chi_{u0}$ and $\chi_{0u}$ only alias with dissipative pairs. When no degree bound is available, a dense variant (\cref{alg:coeff_recovery_dense}) subtracts every potential alias, and the error increases by a factor of $\Kden=O(n^{k-1})$ instead. The full algorithm runs the PTM estimation stage at accuracy $\eta=\eps/\Krec$ with $\Krec=\min\{\Kdeg,\Kden\}$, yielding \cref{thm:learning_upper}. Note that the peeling procedure identifies the significant support during the recursion, since the retained candidates serve as the current hypothesis for the support, so no coefficient-gap or known-support assumption is needed.

\paragraph{Lower bounds.}
For the lower bounds (\cref{sec:lowerbound}), we construct a simple family of hard instances consisting of single-qubit dephasing Lindbladians, whose semigroups are  explicit Pauli channels:
\[
  \cL_\gamma(\rho)=\gamma\ab\big(Z\rho Z-\rho),\qquad
  e^{t\cL_\gamma}(\rho)=(1-p_{\gamma t})\rho+p_{\gamma t}\,Z\rho Z, \qquad p_{\gamma t}=\frac{1-e^{-2\gamma t}}{2},
\]
with $\gamma_0=\Lambda/4$ and $\gamma_1=\gamma_0+4\eps$. Both instances are $1$-local with $\Ddeg=1$ and local dynamical strength $2\gamma_i\leq\Lambda$, and any $\eps$-accurate learner must distinguish them. In other words, each experiment with evolution time $t$ applies $Z$ with probability $p_{\gamma t}$. Revealing whether $Z$ was applied can only help the learner, reducing each hypothesis to a Bernoulli random variable. By the data-processing inequality, each experiment thus increases the relative entropy between the two hypotheses by at most the binary relative entropy $D_\rmb(p_{\gamma_0t}\parallel p_{\gamma_1t})$, which an elementary computation bounds in two ways, uniformly over $t\geq0$:
\[
  D_\rmb\ab\big(p_{\gamma_0t} \parallel p_{\gamma_1t}) \leq \min\cbra[\bigg]{\frac{\Delta^2}{\gamma_0^2},\ \frac{\Delta^2}{\gamma_0}\,t}, \qquad
  \Delta\coloneq\gamma_1-\gamma_0=4\eps.
\]
The two bounds have a clear physical interpretation. For small $t$, the two dephasing probabilities differ by $O(\Delta t)$ while their variance is of order $\gamma_0 t$, so a single experiment provides at most $\Delta^2t/\gamma_0$ information, growing only linearly in $t$ --- in contrast with Hamiltonian parameters, which enter the dynamics as phases and allow quadratic information growth, leading to Heisenberg scaling. For large $t$, both channels converge to the same completely dephasing channel, and the information per experiment is at most $\Delta^2/\gamma_0^2$. By the chain rule and the data-processing inequality, these bounds for individual experiments can be summed over the whole protocol, even when the evolution times, ancillas, and measurements are chosen adaptively (\cref{thm:lower_bound_measurement,thm:lower_bound_time,cor:expected_time_lower}). Combining the resulting bound with Pinsker's inequality implies that the experiment count $N$ and the total evolution time $\Ttot$ of any successful learner must obey
\[
  N = \Omega\ab\Big(\frac{\gamma_0^2}{\Delta^2})=\Omega\ab\Big(\frac{\Lambda^2}{\eps^2}), \qquad
  \Ttot = \Omega\ab\Big(\frac{\gamma_0}{\Delta^2})=\Omega\ab\Big(\frac{\Lambda}{\eps^2}),
\]
which matches the upper bounds up to logarithmic factors and shows that the standard quantum limit cannot be avoided when the dissipative coefficients have to be estimated.

\subsection{Discussion}\label{sec:discussion}
We have presented a near-optimal algorithm for learning local Lindbladians with low-intersection dissipators, whose cost is governed by the local strength $\Lambda$ and the precision $\eps$, together with matching lower bounds that pin the total evolution time to the standard quantum limit $1/\eps^2$. Several directions remain open.

\begin{enumerate}
    \item \textbf{Learning generic $k$-local Lindbladians.} Our algorithm assumes that the dissipative site degree $\Ddeg$ is known. Without knowledge of this parameter, it uses a total evolution time of $\widetilde{O}(n^{2k-2}\Lambda/\eps^2)$; similar guarantees were obtained in concurrent work~\cite{mobus2026robust,lewis2026learning}. In this regime, a polynomial gap remains between this upper bound and our lower bound of $\Omega(\Lambda/\eps^2)$. It is unclear which side is loose: whether an improved algorithm can avoid the $n^{2k-2}$ overhead, or whether a stronger lower bound with explicit dependence on $n$ and $k$ holds.
    \item \textbf{Learning general Lindbladians.} In the ansatz-free setting, the protocol of Ivashkov et al.~\cite{ivashkov2026ansatz} learns $M$-sparse Lindbladians in situ but requires $\widetilde{O}(M^4\nu^2/\eps^4)$ experiments, where $\nu$ is an instance-dependent conditioning factor. Alternative methods eliminate this factor and reach the standard quantum limit, but they require resources beyond the in situ setting, such as ancillary qubits~\cite{romanov2026learning}. It remains open whether a protocol can achieve the standard quantum limit in situ, without any dependence on $\nu$.

    \item \textbf{The gap between the Heisenberg and standard quantum limits.} A Hamiltonian can be learned at the Heisenberg limit $1/\eps$, and merely \emph{detecting} dissipation is also Heisenberg-limited~\cite{cai2026optimal}. By contrast, our lower bound shows that estimating even a single dissipative rate requires total evolution time $\Omega(1/\eps^2)$.
    This gap raises a finer question: which functionals or properties of $\cL$ are learnable at the Heisenberg limit, and which are inherently constrained to the standard quantum limit? More broadly, our lower bound is reminiscent of results in quantum metrology showing that generic noise destroys the Heisenberg scaling of phase estimation~\cite{escher2011general,demkowicz-dobrzanski2012elusive}. Making this analogy quantitative for Lindbladian learning would be interesting.
    \item \textbf{Beyond the Markovian, time-independent setting.} Our results assume Markovian dynamics generated by a time-independent Lindbladian. Extending the guarantees to time-dependent generators, non-Markovian dynamics, or continuous-variable systems would broaden the applicability of the approach.
\end{enumerate}

\section{Background}\label{sec:background}

\subsection{Notation}
Throughout this paper, $\log$ denotes the base-2 logarithm, $\ln$ denotes the natural logarithm, $\iu = \sqrt{-1}$, and $[n] = \{1,2,\ldots,n\}$. Asymptotic notation $O(\cdot)$, $\Theta(\cdot)$, and $\Omega(\cdot)$ has its standard meaning. The notation $\widetilde{O}(\cdot)$ and $\widetilde{\Theta}(\cdot)$ suppress polylogarithmic factors.

Matrix norms in this paper are Schatten norms: for a matrix $A$, let $\norm{A}_{p}$ denote its Schatten-$p$ norm.  In particular, $\norm{A}_{1}$ is the trace norm and $\norm{A}_{\infty}$ is the operator norm. For a linear map $\Phi$ on matrices, define the induced operator norm as
\[
  \norm{\Phi}_{\infty\to\infty} \coloneq \sup_{A\neq 0}\frac{\norm{\Phi(A)}_{\infty}}{\norm{A}_{\infty}}.
\]
We use $\Phi^\dag$ for the adjoint with respect to the Hilbert--Schmidt inner product.

\subsection{Pauli operators}
We often expand operators in the basis of tensor products of Pauli matrices.
\begin{definition}[Pauli matrices]
    The Pauli matrices are $2 \times 2$ Hermitian matrices defined as follows.
\[
    I = \begin{pmatrix} 1 & 0 \\ 0 & 1 \end{pmatrix}, \quad
    X = \begin{pmatrix} 0 & 1 \\ 1 & 0 \end{pmatrix}, \quad
    Y = \begin{pmatrix} 0 & -\iu \\ \iu & 0 \end{pmatrix}, \quad
    Z = \begin{pmatrix} 1 & 0 \\ 0 & -1 \end{pmatrix}.
\]
\end{definition}
These matrices are unitary and Hermitian, and $XY = \iu Z, YZ = \iu X, ZX = \iu Y$. Therefore, the product of Pauli matrices is a Pauli matrix, up to a factor of $\{1,\iu,-1,-\iu\}$. The non-trivial (non-identity) Pauli matrices are traceless. We consider tensor products of Pauli matrices, $P_1 \otimes \cdots \otimes P_n$, where $P_i \in \{X,Y,Z,I\}$ for all $i \in [n]$. Let $\cP_n$ denote the set of $n$-qubit Pauli operators. For both notation and computation, it is convenient to encode each Pauli operator by a binary vector, which turns operator multiplication into addition modulo two.

\begin{definition}[Binary representation of Pauli operators]
    Given a Pauli operator $P_a$, define the binary vector $a=(z,x)\in \F_{2}^{2n}$ as the binary representation (label) of $P_a$ if
    \[
        P_a=\bigotimes_{i=1}^n (P_a)_i=(-\iu)^{z\cdot x}\bigotimes_{i=1}^n Z^{z_i}X^{x_i},
    \]
    where each $z_i,x_i\in \{0,1\}$.
\end{definition}
The zero vector $0\in\F_2^{2n}$ labels the identity.
We will often identify $P_a$ with its label $a$ when no confusion can arise. Two basic quantities associated with a Pauli operator are its support and weight, recording the qubits on which it acts nontrivially and how many there are.
\begin{definition}[Support and weight of a Pauli operator]
    For a Pauli operator $P_a \in \cP_n$, its support $\supp(a)~\subseteq~[n]$ is the subset of qubits that $P_a$ acts non-trivially on. That is,
    \[
        \supp(a)\coloneq \{i:(P_a)_i\neq I\}.
    \]
    The weight of the Pauli operator $P_a$ is the size of the support, denoted by
    \[
        \wt(a)\coloneq \abs{\supp(a)}.
    \]
\end{definition}
These notions extend naturally to a pair of Paulis. For a pair of Pauli labels $(a,b)$, write
\[
  \supp(a,b)\coloneq \supp(a)\cup\supp(b), \qquad \wt(a,b)\coloneq \abs{\supp(a,b)}.
\]

Bounding the weight yields the notion of locality that underlies our setting.

\begin{definition}[Local Pauli operators]
For $k\geq 0$, define the set of all non-identity $k$-local Pauli operators as
\[
  \cP_{n,k}\coloneq \{P_a\in\cP_n:1\leq\wt(P_a)\leq k\}.
\]
\end{definition}
For a subset $S \subseteq [n]$ and a label $a=(z,x)\in\bits^{2n}$, we write $a \in \bits^{2S}$ if $P_a$ is supported on $S$, i.e., $\supp(a) \subseteq S$. We write $a|_S$ for the sub-operator of $a$ within the region $S$ obtained by replacing every operator
outside $S$ by the identity,
\[
  a|_S \coloneq (z',x')\in\bits^{2n}, \qquad
  (z'_i,x'_i)=
  \begin{cases}
    (z_i,x_i), & i\in S,\\
    (0,0), & i\notin S.
  \end{cases}
\]

Apart from supports, we will repeatedly use the commutation structure of Pauli operators, which is conveniently captured by the symplectic inner product of their labels.

\begin{definition}[Symplectic inner product]
    Given two Pauli labels $a,b \in \bits^{2n}$, we define their symplectic inner product as
    \[
      \langle a,b\rangle_s = z\cdot x' + x\cdot z' \pmod 2,
    \]
    where $a=(z,x)$ and $b=(z',x')$.
\end{definition}

The Lindbladian generator is built from commutators and anti-commutators, which we recall next.
\begin{definition}[Commutator and anti-commutator]
    Given operators $A$ and $B$, the commutator of $A$ and $B$ is defined as 
    \[
        [A,B] = AB - BA,
    \]
    and the anti-commutator is defined as
    \[
        \{A,B\} = AB + BA.
    \]
\end{definition}
The symplectic inner product between two Pauli operators can be regarded as an indicator of commutation. Given two Pauli operators $P_a,P_b \in \cP_n$, we have
\[
    P_aP_b=(-1)^{\langle a,b\rangle_s}P_bP_a.
\]
Pauli operators are moreover closed under multiplication, up to a phase, which lets us track products purely at the level of binary labels. Given $P_a,P_b \in \cP_n$, their multiplication has the form
\[
  P_aP_b = \xi_{ab}P_{a\oplus b},
\]
where $\xi_{ab}\in\{\pm 1,\pm \iu\}$ are efficiently computable given $a$ and $b$, and satisfy $\xi_{ab}\xi_{ba}=1$. Moreover, $\xi_{ab} \in \{\pm 1\}$ if $P_a$ and $P_b$ commute, and $\xi_{ab} \in \{\pm \iu\}$ if they anti-commute.

\subsection{Markovian open quantum systems}

For an open quantum system, when environmental memory effects are negligible, the reduced dynamics becomes Markovian and is generated by a Lindbladian \cite{gorini1976completely, lindblad1976generators}. A Lindbladian naturally separates the evolution into two contributions: a coherent part generated by an effective Hamiltonian, and dissipative parts that encode irreversible environmental effects such as dephasing and relaxation.

\begin{definition}[Quantum dynamical semigroup]\label{def:semigroup}
Let $\mathcal{H}$ be an $n$-qubit Hilbert space and $\mathcal{B}(\mathcal{H})$ be the set of all linear operators on $\mathcal{H}$.
A family of linear maps
$\{\mathcal{E}_t\}_{t\geq 0}$, with
$\mathcal{E}_t\colon\mathcal{B}(\mathcal{H})\to\mathcal{B}(\mathcal{H})$,
is called a \emph{quantum dynamical semigroup} if
\begin{itemize}
    \item $\mathcal{E}_0=\mathrm{id}$, the identity map on $\mathcal{B}(\mathcal{H})$,
    \item $\mathcal{E}_{t+s}=\mathcal{E}_t\circ\mathcal{E}_s$ for all $t,s\geq 0$,
    \item each $\mathcal{E}_t$ is a quantum channel, i.e.\ a completely positive and trace-preserving map,
    \item $t\mapsto\mathcal{E}_t$ is continuous.
\end{itemize}
\end{definition}

In finite dimensions, the Gorini--Kossakowski--Sudarshan--Lindblad (GKSL) theorem \cite{gorini1976completely, lindblad1976generators} provides the closed form of the generator for this semigroup.

\begin{definition}[Lindbladian]
    A Lindbladian is the infinitesimal generator of a quantum dynamical semigroup, defined by
    \[
    \mathcal{L} = \lim_{t\to 0^+} \frac{\mathcal{E}_t-\mathrm{id}}{t},
    \]
    which can be written in the Pauli basis as
    \begin{equation*}
      \cL(\rho) = \sum_{a\neq 0} -\iu h_a[P_a,\rho] + \sum_{b,c\neq 0} \gamma_{bc} \ab\Big(P_b\rho P_c-\frac{1}{2}\{P_cP_b,\rho\}), \qquad h_a \in \bbR \text{ and }\gamma_{bc} \in \bbC,
    \end{equation*}
where $\gamma=(\gamma_{bc})_{b,c\neq0} \succeq 0$ and $\gamma=\gamma^\dagger.$
\end{definition}

The evolution of the Lindbladian can be written as $\cE_t = e^{t\cL}$, and by definition we have
\begin{align*}
    \odv{\rho(t)}{t}=\mathcal{L}(\rho(t)).
\end{align*}

In this paper, we consider Lindbladians that are $k$-local. Since the dissipation contains two-sided Pauli terms, we define the set of Pauli pairs of weight at most $k$,
\begin{equation}
    \Gamma_k \coloneq \ab\{(b,c)\in\bits^{2n}\times\bits^{2n}:b\neq0,\ c\neq0,\ 1\leq \wt(b,c)\leq k\}.
\end{equation}

\begin{definition}[Local Lindbladian]\label{def:local_lindbladian}
    A Lindbladian is $k$-local if it has the form 
    \begin{equation}\label{eq:lindbladian}
      \cL(\rho) = \sum_{a\in\cP_{n,k}}-\iu h_a [P_a,\rho] + \sum_{(b,c)\in \Gamma_k} \gamma_{bc} \ab\Big(P_b\rho P_c-\frac{1}{2}\{P_c P_b,\rho\}).
    \end{equation}
    In other words, each coherent and dissipative term acts nontrivially on at most $k$ qubits.
\end{definition}

\begin{definition}[Dissipative site degree]
    The dissipative site degree of a Lindbladian is the maximum number of dissipative Pauli-pair terms acting non-trivially on any single qubit,
    \[
      \Ddeg \coloneq \max_{i\in[n]}{\abs{\{(b,c)\in\Gamma_k:\gamma_{bc}\neq0,\ i\in\supp(b,c)\}}}.
    \]
\end{definition}
The dissipative site degree is bounded if $\Ddeg = O(1)$. For a Lindbladian with bounded dissipative site degree, we say it has \emph{low-intersection dissipators}.

A natural measure of the overall strength of the dynamics is the induced operator norm of the generator under the Heisenberg picture, $\norm{\cL^\dag}_{\infty\to\infty}$. This global quantity generically grows with the system size: an extensive Lindbladian with $O(1)$ local coefficients is a sum of $\Theta(n)$ local terms, so the norm scales as $\Theta(n)$. The complexity of our algorithm is instead governed by a \emph{local} strength, which we now define.

\begin{definition}[Local dynamical strength]\label{def:local_strength}
    Decompose a Lindbladian into local generators $\cL = \sum_{S} \cL_S$, where $S$ is a local region such that $\cL_S$ collects all Hamiltonian and dissipative terms supported exactly on $S$, 
    \[
    \cL_S(\rho) = \sum_{\supp(a)= S}-\iu h_a [P_a,\rho] + \sum_{\supp(b, c)= S} \gamma_{bc} \ab\Big(P_b\rho P_c-\frac{1}{2}\{P_c P_b,\rho\}).
    \]
    The \emph{local dynamical strength} under the Heisenberg picture is the largest total interaction strength attached to any single qubit,
    \[
    \Lloc \coloneq \max_{i \in [n]}{\sum_{S \ni i}{\norm[\big]{\cL_S^\dagger}}_{\infty\to\infty}}.
    \]
\end{definition}

Since $\Lloc$ only counts the terms acting on a single qubit, the global dynamical strength obeys $\norm{\cL^\dag}_{\infty\to\infty}\leq\sum_S\norm{\cL_S^\dag}_{\infty\to\infty}=O(n\Lloc)$, so the local strength can be smaller than the global one by a factor of the system size. Throughout the paper, we assume the local strength is bounded by a known parameter $\Lambda$, that is, $\Lloc\leq\Lambda$. We will show that both the experiment count and the total evolution time depend on $\Lambda$ rather than on the global dynamical strength.

This model includes many physically and practically relevant cases. For example, consider a one-dimensional transverse-field Ising chain, $H=\sum_i J_i Z_iZ_{i+1}+\sum_i h_i X_i$, coupled to a local thermal bath through single-site relaxation and excitation operators $\sigma_i^\pm=(X_i\pm \iu Y_i)/2$. The whole dynamics is $2$-local and the dissipative site degree is at most 4. The same setting also captures common noise models for quantum hardware. For instance, in a superconducting-qubit device with bounded-degree
nearest-neighbor connectivity, the Hamiltonian contains local control terms and two-qubit couplings along hardware edges, while the noise is dominated by single-qubit relaxation, excitation, dephasing, and possible two-qubit correlated dephasing on neighboring qubits. These processes can be described by Lindbladians that are at most $2$-local with bounded dissipative site degree when the hardware connectivity is bounded.

\subsection{Pauli transfer matrix}
We define the Pauli transfer matrix (PTM) of the Lindbladian.
\begin{definition}[Pauli transfer matrix of Lindbladian]
    For a Lindbladian $\cL$ and time $t \geq 0$, the continuous-time Pauli transfer matrix of the semigroup $e^{t\cL}$ has entries
    \begin{equation}
        R_{vw}(t) \coloneq \frac{1}{2^n} \tr\ab\big(P_v e^{t \cL}(P_w)), \qquad P_v, P_w \in \cP_n.
    \end{equation}
    Its generator (the PTM of the Lindbladian $\cL$ itself) has entries
    \begin{equation}
        L_{vw} \coloneq \odv[delims-eval=.|]{R_{vw}(t)}{t}_{t=0} = \frac{1}{2^n} \tr\ab\big(P_v \cL(P_w)).
    \end{equation}
\end{definition}
We will estimate all entries $L_{vw}$ for $P_v$ and $P_w$ whose joint support has size at most $k$, and show that this is sufficient to recover the coefficients of a $k$-local Lindbladian. Intuitively, the recovery is carried out locally, region by region, so only such entries are ever needed. We collect these target entries into the index set
\begin{equation}\label{eq:Tk_def}
  \cT_{k}\coloneq\{(v,w):\wt(v, w)\leq k\},
\end{equation}
and write $M_k\coloneq\abs{\cT_k}=O(n^{k})$ for its size.

This index set $\cT_k$ resembles the dissipative index set $\Gamma_k$, as both collect Pauli pairs of joint weight at most $k$ and have size $O(n^k)$. The difference lies in what they describe: $\cT_k$ indexes the PTM entries $L_{vw}$ that we estimate from the dynamics, while $\Gamma_k$ indexes the coefficients $\gamma_{bc}$ that we eventually output. In addition, $\cT_k$ keeps the entries in which one Pauli is the identity, which are the ones carrying the information needed to recover the Hamiltonian coefficients.

\subsection{Derivative estimation by Chebyshev interpolation}\label{sec:chebyshev_interpolation}
For $q\in \N_{>0}$, let $f(x)$ be a real-valued function that is $q+1$ times differentiable on a closed interval $\cI$. Let $x_0, x_1, \dots, x_q$ be $q+1$ distinct interpolation nodes in $\cI$, and let $p_q(x)$ be the unique polynomial of degree at most $q$ that interpolates $f(x)$ at these nodes (i.e., $p_q(x_j) = f(x_j)$ for all $j = 0, 1, \dots, q$).

\begin{lemma}[Interpolation remainder formula, {see e.g.\ \cite[Chapter 3]{burden2011numerical}}]\label{lem:remain_formula}
    For any $x \in \cI$, there exists a point $\varsigma$ situated in the smallest interval containing $x, x_0, x_1, \dots, x_q$ such that the interpolation error $f(x) - p_q(x)$ is given by
    \[
        f(x) - p_q(x) = \frac{f^{(q+1)}(\varsigma)}{(q+1)!} \prod_{j=0}^q (x - x_j).
    \]
\end{lemma}

We will only apply the interpolation-error formula at the node $x=x_0$, where the standard differentiated form gives, for some $\varsigma$ in the interval,
\[
  f'(x_0) - p_q'(x_0) = \frac{f^{(q+1)}(\varsigma)}{(q+1)!} \prod_{j\neq 0}^q(x_0-x_j),
\]
which avoids differentiating the (possibly non-smooth) remainder point $\varsigma$ in $x$.

We use Chebyshev--Lobatto interpolation. It operates on the Chebyshev--Lobatto nodes, which are clustered near the boundaries of the domain $[-1, 1]$.
\begin{definition}[Chebyshev--Lobatto nodes]
    Let $q \in \N_{>0}$. For a grid of $q+1$ points, the nodes are defined as
    \[
        x_j = \cos\ab\Big(\frac{j\pi}{q}), \quad j=0,1,\dots,q.
    \]
\end{definition}
\begin{definition}[Chebyshev spectral differentiation matrix]
Given $q \in \N_{>0}$ and a polynomial $f(x)$ of degree at most $q$, the evaluation of this polynomial at $q + 1$ Chebyshev--Lobatto nodes is a vector
\[
    \mathbf{f} = [f(x_0),f(x_1),\dots,f(x_q)]^\top.
\]
We define the Chebyshev spectral differentiation matrix as the $(q+1)\times(q+1)$ matrix $D$ such that multiplying it by $\mathbf{f}$ yields the vector of derivatives $\mathbf{f}'$, that is,
\[
 (D\mathbf{f})_j = \mathbf{f}'_j = f'(x_j),
\]
for $j \in \{0,1,\dots,q\}$. For a general function $f$, the product $D\mathbf{f}$ returns the derivatives at the nodes of its degree-$q$ interpolant.
\end{definition}
\begin{lemma}[{\cite[Section 2.4.2]{canuto2006spectral}}]\label{lem:chebyshev_matrix}
    For $q \in \N_{>0}$, the Chebyshev spectral differentiation matrix $D = D_q$ has the following entries
    \begin{gather*}
        D_{00}=\frac{2q^2+1}{6}, \quad D_{qq}=-\frac{2q^2+1}{6}, \\
        D_{ii}=\frac{-x_i}{2(1-x_i^2)}\quad\text{for } 1 \leq i \leq q-1, \\
        D_{ij}=\frac{c_i}{c_j}\frac{(-1)^{i+j}}{x_i-x_j}\quad\text{for } i \neq j,
    \end{gather*}
    where the coefficients $c_i$ are defined such that $c_0=c_q=2$, and $c_i=1$ for all interior points $1\leq i \leq q-1$.
\end{lemma}

\subsection{Information-theoretic quantities}
We collect the information-theoretic notions used in the lower-bound arguments. The notation $D(\cdot\parallel\cdot)$ is used for both classical and quantum relative entropy. The type of the inputs will make clear which notion is meant.
\begin{definition}[Classical relative entropy]\label{def:classical_relative_entropy}
For probability distributions $\mu,\nu$ on a finite set, the relative entropy, also called the Kullback--Leibler divergence, is
\[
  D(\mu\parallel \nu) \coloneq \sum_x \mu(x)\ln\frac{\mu(x)}{\nu(x)},
\]
with the convention that $D(\mu\parallel \nu)=\infty$ if $\mu$ is not absolutely continuous with respect to $\nu$.  For Bernoulli distributions with parameters $p,q\in[0,1]$, we write the binary relative entropy as
\[
  D_\rmb(p\parallel q) \coloneq p\ln\frac{p}{q} + (1-p)\ln\frac{1-p}{1-q}.
\]

\end{definition}
The quantum analogue is defined on density operators and is the quantity used when comparing the possible quantum states produced by two candidate dynamics.
\begin{definition}[Quantum relative entropy]\label{def:quantum_relative_entropy}
For quantum states $\rho,\sigma$, the quantum relative entropy, or quantum Kullback--Leibler divergence, is
\[
  D(\rho\parallel\sigma) \coloneq \tr\ab(\rho(\ln\rho-\ln\sigma)),
\]
when $\supp(\rho)\subseteq\supp(\sigma)$, and is infinite otherwise.
\end{definition}
To convert distinguishability statements into error-probability lower bounds, we also use the total variation distance between classical outcome distributions.
\begin{definition}[Total variation distance]\label{def:total_variation}
For probability distributions $\mu,\nu$ on a finite set, the total variation distance is
\[
  \dtv(\mu,\nu) \coloneq \frac{1}{2}\sum_x\abs{\mu(x)-\nu(x)}.
\]
\end{definition}
The following standard facts are used to control how much information can be gained in each experiment.
\begin{fact}[Data-processing inequality]\label{fact:data_processing}
Relative entropy is monotone under stochastic maps and quantum channels.  In particular, for any quantum channel $\Phi$ and quantum states $\rho,\sigma$ satisfying $\supp(\rho)\subseteq\supp(\sigma)$,
\[
  D(\Phi(\rho)\parallel\Phi(\sigma)) \leq D(\rho\parallel\sigma),
\]
and the analogous inequality holds for classical probability distributions under stochastic maps.
\end{fact}

\begin{fact}[$\chi^2$-divergence bound on binary relative entropy]\label{fact:chi2}
For Bernoulli distributions with parameters $p,q\in(0,1)$, the binary relative entropy satisfies
\[
  D_\rmb(p\parallel q)\leq \frac{(p-q)^2}{q(1-q)}.
\]
\end{fact}

\begin{fact}[Pinsker's inequality]\label{fact:pinsker}
For probability distributions $\mu,\nu$ on a finite set,
\[
  \dtv(\mu,\nu) \leq \sqrt{\frac{1}{2}D(\mu\parallel\nu)}.
\]
\end{fact}

\section{Lindbladian learning algorithm and analysis}\label{sec:algorithm}

This section presents and analyzes our learning algorithm. We first estimate the PTM generator entries from the dynamics: \cref{sec:chebyshev} sets up the endpoint-derivative rule via Chebyshev interpolation, \cref{sec:ptm_learning} estimates the finite-time PTM entries by shadow process tomography, and \cref{sec:upper_bound} bounds the resulting sample and time complexity. We then convert the estimated PTM generator into the physical coefficients in \cref{sec:coeff_recovery}, and combine the two stages into our main upper bound in \cref{thm:learning_upper}. Throughout this section, we assume $\eps < k\Lambda$. Otherwise, outputting the all-zero Lindbladian suffices, since every coefficient of $\cL$ is bounded in magnitude by $\Lambda \leq \eps$. This assumption ensures $q \geq 1$ in \cref{lem:truncation}.

\subsection{Chebyshev interpolation}\label{sec:chebyshev}

The dynamic-access model gives estimates of finite-time PTM entries $R_{vw}(t)$, while the desired generator entry is the endpoint derivative $L_{vw}=R_{vw}'(0)$. Thus, the learning problem contains a differentiation step: we must approximate a derivative at the boundary of the interval from estimates at positive times. We use Chebyshev--Lobatto interpolation for this step, which gives an explicit endpoint derivative rule and has a rapidly decaying interpolation bias for the analytic semigroup functions considered here. Chebyshev--Lobatto endpoint differentiation is a standard spectral interpolation tool~\cite{howell1991derivative,canuto2006spectral}.  Similar interpolation ideas have also been used in recent algorithms for learning quantum dynamical generators~\cite{caro2024learning,franca2024efficient,gu2024practical,ivashkov2026ansatz,he2026optimal}.

We have the following bound on the $r$-th derivative of local entries $R_{vw}(t)$, which depends on the local dynamical strength.

\begin{lemma}\label{lem:derivative_bound}
    Let $\cL$ be a $k$-local Lindbladian with local dynamical strength $\Lloc$. For any $(v,w)\in \cT_k$, $r \in \N$ and $t\geq 0$, the $r$-th derivative of PTM entry $R_{vw}(t)$ satisfies
    \[
     \abs[\big]{R_{vw}^{(r)}(t)} \leq  r!\,(k\Lloc)^r.
    \]
\end{lemma}

\begin{proof}
    By differentiating the semigroup,
    \[
    \odv[r]{R_{vw}(t)}{t}=\frac{1}{2^n}\tr\ab\big((\cL^\dag)^r(P_v)\,e^{t\cL}(P_w)).
    \]
    Since $e^{t\cL}$ is a quantum channel, the trace norm is contractive on Hermitian operators, i.e., $\norm{e^{t\cL}(P_w)}_{1}\leq \norm{P_w}_{1}=2^n$. By H\"{o}lder's inequality for Schatten norms,
    \[
    \abs[\Big]{\odv[r]{R_{vw}(t)}{t}} \leq \frac{1}{2^n}\norm[\big]{(\cL^\dag)^r(P_v)}_{\infty}
    \norm[\big]{e^{t\cL}(P_w)}_{1} \leq \norm[\big]{(\cL^\dag)^r(P_v)}_{\infty}.
    \]
    It remains to bound $\norm{(\cL^\dag)^r(P_v)}_{\infty}$ by the local strength. Using the decomposition $\cL^\dag=\sum_S\cL_S^\dag$ from \cref{def:local_strength},
    \[
      (\cL^\dag)^r(P_v)=\sum_{S_1,\dots,S_r}\cL_{S_r}^\dag\cdots\cL_{S_1}^\dag(P_v).
    \]
    Each $\cL_S^\dag$ is supported on $S$, annihilates the identity, and is bounded as a map. Hence $\cL_S^\dag(A)=0$ whenever $\supp(A)\cap S=\varnothing$, and otherwise $\norm{\cL_S^\dag(A)}_{\infty}\leq b_S\norm{A}_{\infty}$ with $b_S\coloneq\norm{\cL_S^\dag}_{\infty\to\infty}$ and $\supp(\cL_S^\dag(A))\subseteq\supp(A)\cup S$. Set $U_0=\supp(v)$ and $U_\ell=U_{\ell-1}\cup S_\ell$. A summand is nonzero only if $S_\ell\cap U_{\ell-1}\neq\varnothing$ for every $\ell$, in which case its operator norm is at most $\prod_{\ell=1}^r b_{S_\ell}$. Bounding the nested sum from the inside out, and using $\abs{U_{\ell-1}}\leq\wt(v)+(\ell-1)k\leq \ell k$ together with $\sum_{S\ni i}b_S\leq\Lloc$,
    \[
      \norm[\big]{(\cL^\dag)^r(P_v)}_{\infty} \leq\prod_{\ell=1}^r\ab\Big(\sum_{S: S\cap U_{\ell-1}\neq\varnothing}b_S)
      \leq\prod_{\ell=1}^r\abs{U_{\ell-1}}\,\Lloc
      \leq\prod_{\ell=1}^r \ell\, k\,\Lloc
      = r!\,(k\Lloc)^r. \qedhere
    \]
\end{proof}

Let $q \in \N_{>0}$ and $T > 0$. We map the $q+1$ standard Chebyshev--Lobatto nodes from the canonical domain $[-1, 1]$ to the physical time interval $[0, T]$:
\[t_j = \frac{T}{2} \ab( 1 - \cos\ab\Big(\frac{j\pi}{q})), \quad j = 0, 1, \dots, q.\]
Note that $t_0 = 0$ and $t_q = T$. The true expectation value is $R_{vw}(t)$, and let $p_q(t)$ be the exact interpolating polynomial. The derivative of the interpolating polynomial at $t = 0$ is a linear combination of $R_{vw}(t_j)$,
\begin{equation*}
    p_q'(0) = \sum_{j=0}^q \ell'_j(0) R_{vw}(t_j),
\end{equation*}
where $\ell_j(t)$ are the Lagrange basis polynomials
\[
\ell_j(t) = \prod_{\substack{0\leq i \leq q\\ i \neq j}} \frac{t - t_i}{t_j - t_i}.
\]
\begin{lemma}[Interpolation error]\label{lem:truncation}
    Let $\cL$ be a $k$-local Lindbladian with local dynamical strength $\Lloc \leq \Lambda$. For any $\eps > 0$, set the evolution time and the number of nodes as
    \[
        T = \frac{1}{2k\Lambda},\qquad
        q = \ab\lceil\log(2k\Lambda/\eps)\rceil.
    \]
    Then, for every $(v,w)\in\cT_k$, the interpolation error is bounded by
    \[
         \abs{L_{vw} - p_q'(0)} \leq \frac{\eps}{2}.
    \]
\end{lemma}
\begin{proof}
    By the differentiated interpolation-remainder formula in \cref{lem:remain_formula} applied at the node $x_0 = t_0 = 0$, there exists $\varsigma_0 \in (0, T)$ such that
    \[
        L_{vw} - p_q'(0) = \frac{R_{vw}^{(q+1)}(\varsigma_0)}{(q+1)!} \prod_{j\neq 0}(0 - t_j) = \frac{R_{vw}^{(q+1)}(\varsigma_0)}{(q+1)!} \cdot (-1)^q \prod_{j=1}^q t_j.
    \]
    Here we used $t_0=0$, so that differentiating $W(t)=\prod_{j=0}^q(t-t_j)$ at $t=0$ leaves only the term $\prod_{j\neq0}(0-t_j)$ and avoids differentiating the $t$-dependent remainder point $\varsigma_t$.
    By the local derivative bound in \Cref{lem:derivative_bound} and the assumption $\Lloc\leq\Lambda$, $\abs{R_{vw}^{(q+1)}(\varsigma_0)}\leq (q+1)!\,(k\Lloc)^{q+1}\leq (q+1)!\,(k\Lambda)^{q+1}$, and since every node obeys $t_j\leq T$,
    \begin{align*}
        \abs{L_{vw} - p_q'(0)} &= \frac{\abs{R_{vw}^{(q+1)}(\varsigma_0)}}{(q+1)!} \cdot \prod_{j=1}^q t_j \\
        & \leq \frac{(q+1)!\,(k\Lambda)^{q+1}}{(q+1)!}\,T^{q} = (k\Lambda)^{q+1}T^q.
    \end{align*}
    Fixing the maximum evolution time as $T = 1/(2k\Lambda)$, the factorials cancel and the truncation error decays geometrically in $q$,
    \[
        (k\Lambda)^{q+1}T^q = k\Lambda\,\ab\big(k\Lambda T)^q = k\Lambda\,2^{-q}.
    \]
    To bound this by $\eps/2$, it suffices to fix the number of nodes as
    \begin{equation*}
        q = \ab\lceil\log(2k\Lambda/\eps)\rceil = O\ab\big(\log(k\Lambda/\eps)) = \widetilde{\Theta}(1). \qedhere
    \end{equation*}
\end{proof}

\subsection{Learning the PTM generator by shadow process tomography}\label{sec:ptm_learning}
The learning protocol below is adapted from ancilla-free shadow process tomography and related randomized-measurement methods~\cite{franca2024efficient,levy2024classical,kunjummen2023shadow}.  We specialize it to the local PTM entries needed for the endpoint-differentiation step. 

The target PTM entries are those indexed by $\cT_k$ defined in \cref{eq:Tk_def}. In the generator estimation algorithm, we choose the Chebyshev--Lobatto nodes $t_0,\ldots,t_q$ as described in \cref{sec:chebyshev_interpolation}, and estimate each $R_{vw}(t_j)$ for $0 \leq j \leq q$.  The value at $t_0=0$ is known exactly, since $R_{vw}(0)=\mathbf{1}[v=w]$, so no experiment is performed at $t_0$.  For each nonzero node $t_j$, $j=1,\ldots,q$, we collect $N$ independent shadow records.

In each experiment for some fixed $t$, we first generate a random input state $\rhoin = \bigotimes_{i=1}^n \ketbra{\psi_i}{\psi_i}$ where each $\ket{\psi_i}$ is chosen independently and uniformly from $\cbra{\ket{0},\ket{1},\ket{+},\ket{-},\ket{+\iu},\ket{-\iu}}$. After applying $e^{t\cL}$, choose a measurement basis $B_i\in\{X,Y,Z\}$ independently and uniformly for each output qubit, and record outcomes $o_i\in\{\pm1\}$. A single shadow shot can be recorded as a triple
\[
s = (\rhoin, B, o),
\]
where $B = (B_1, B_2, \dots, B_n)$ is the vector of output measurement bases and $o = (o_1, o_2, \dots, o_n)$ is the vector of measurement outcomes. Thus, the full shadow data set contains $qN$ raw records collected in exactly $qN$ experiments. For each target entry $(v,w)$, every raw record contributes one bounded random variable to the average for $R_{vw}(t)$.  The same collection of records is reused for all entries in $\cT_k$.

For a node $t_j$, the estimator of PTM entries $(v,w)$ obtained from the $m$-th shot is defined as
\begin{align}
\hat{R}_{vw}(t_j,m) \coloneq \hat{C}_{\mathrm{in}}(w) \cdot\hat{C}_{\mathrm{out}}(v),
\end{align}
where
\[
\hat{C}_{\mathrm{in}}(w) \coloneq 3^{\wt(w)}\tr\ab\big(P_w \rhoin), \qquad \hat{C}_{\mathrm{out}}(v) \coloneq \prod_{i\in\supp(v)} 3 o_i \cdot \mathbf{1}[B_i = (P_v)_i].
\]
The empirical PTM estimate at time $t_j$ is the average over total $N$ experiments,
\begin{equation*}
\bar{R}_{vw}(t_j) \coloneq \frac{1}{N}\sum_{m=1}^N \hat{R}_{vw}(t_j,m).
\end{equation*}
The generator estimate is then obtained by applying the endpoint differentiation rule to these empirical finite-time PTM estimates, together with the exact value at $t_0=0$:
\[
  \hat{L}_{vw} = \sum_{j=0}^q \ell'_j(0) \bar{R}_{vw}(t_j).
\]
In the following, we show that each single-shot estimator is unbiased and uniformly bounded.

\begin{algorithm}[ht!]
\caption{Learning the PTM generator by shadow process tomography}
\label{alg:PTM_learning}
\begin{algorithmic}[1]
\Require Access to $n$-qubit channel $\{e^{t\cL}\}_{t\geq 0}$; locality $k$; accuracy $\eps$; failure probability $\delta$; local strength bound $\Lambda$.
\Ensure Estimates of the PTM generator entries $\{\hat{L}_{vw}\}_{(v,w)\in \cT_k}$.
\State Set $T \gets 1/(2k\Lambda)$ and $q \gets \lceil\log(2k\Lambda/\eps)\rceil$.
\State Compute Chebyshev--Lobatto nodes $t_j \in [0, T]$ and weights $\ell'_j(0)$ for $j = 0,1,\dots, q$.
\State Set $N \gets C\, 3^{2k}k^2\Lambda^2 q^4/{\eps^2}\cdot(k\log n + \log(2/\delta))$ for a sufficiently large constant $C$.
\Statex \textit{// Stage 1: collecting raw shadow records by using the channel.}
\For{$j=1,\dots,q$}
  \For{$m = 1,\dots,N$}
    \State Prepare $\rhoin = \bigotimes_{i=1}^n\ketbra{\psi_i}{\psi_i}$, each $\ket{\psi_i}$ chosen uniformly from $\cbra*{\ket{0},\ket{1},\ket{+},\ket{-},\ket{+\iu},\ket{-\iu}}$.
    \State Evolve $\rhoout \gets e^{t_j\cL}(\rho_{\mathrm{in}})$.
    \State Measure each qubit in an independently and uniformly sampled Pauli basis 
    \NoNumState $B_i \in \{X, Y, Z\}$, forming $B = (B_1, \dots, B_n)$, and get an outcome $o\in\{\pm 1\}^n$.
    \State Store the raw record $s_{j,m}\gets(\rhoin,B,o)$.
  \EndFor
\EndFor
\Statex \textit{// Stage 2: classical post-processing, reusing the same collection of records for every $(v,w)$.}
\For{$(v,w) \in \cT_k$}
  \For{$j=1,\dots,q$ \textbf{and} $m=1,\dots,N$}
  \State Read record $s_{j,m}$, compute
  \[
  \hat{R}_{vw}(t_j,m) \gets \hat{C}_{\mathrm{in}}(w) \cdot\hat{C}_{\mathrm{out}}(v)
  \]
  \NoNumState where
  \[
  \hat{C}_{\mathrm{in}}(w) \gets 3^{\wt(w)}\tr\ab\big(P_w \rhoin), \qquad
  \hat{C}_{\mathrm{out}}(v) \gets \prod_{i\in\supp(v)} 3 o_i \cdot \mathbf{1}[B_i = (P_v)_i].
  \]
\EndFor
\State $\bar{R}_{vw}(t_j) \gets \frac{1}{N}\sum_{m=1}^N \hat{R}_{vw}(t_j,m)$ for $j = 1, \dots, q$.
\EndFor

\State\Return
\begin{equation}\label{eq:shadow_deriv}
  \hat{L}_{vw} \gets \ell'_0(0)\,\mathbf{1}[v=w] + \sum_{j=1}^{q}\ell'_j(0) \bar{R}_{vw}(t_j) 
\end{equation}
for all $(v,w)\in \cT_k$.
\end{algorithmic}
\end{algorithm}
\begin{lemma}[Unbiased estimator of PTM]\label{lem:unbiasedness_of_R}
    For every pair $(v, w) \in \cT_k$, $t \geq 0$, and $m \in [N]$,
    \[
    \E[\hat{R}_{vw}(t, m)] = R_{vw}(t).
    \]
\end{lemma}
\begin{proof}
    Let $\rhoout = e^{t\cL}(\rhoin)$. 
    The measurement bases $B = (B_1, B_2, \dots, B_n)$ are chosen independently and uniformly from $\cbra{X, Y, Z}$. Since $P_v = \bigotimes_{i=1}^n (P_v)_i$, the estimator
    \[
    \hat{C}_{\mathrm{out}}(v) = \prod_{i\in\supp(v)} 3 o_i \cdot \mathbf{1}[B_i = (P_v)_i]
    \]
    is non-zero only when $B_i = (P_v)_i$ for every $i \in \supp(v)$, the probability of which is 
    \[
    \Pr\ab\big[B_i = (P_v)_i, \forall i \in \supp(v)] = \ab\Big(\frac{1}{3})^{\wt(v)}.
    \]
    Conditioning on this event, the outcomes $\cbra{o_i}_{i \in \supp(v)}$ are the results of measuring $P_v$ qubit-wise on $\rhoout$, thus we have
    \[
    \E\ab\bigg[\prod_{i \in \supp(v)} o_i \mathrel{\Big|} B_i = (P_v)_i, \forall i \in \supp(v)] = \tr(P_v \rhoout).
    \]
    Therefore,
    \begin{subequations}
    \label{eq:C_out}
    \begin{align}
        \E_{B,o}\ab[\hat{C}_{\mathrm{out}}(v) \mathrel{\big|} \rhoin] &= 3^{\wt(v)} \cdot \tr(P_v \rhoout) \cdot \Pr\ab\big[B_i = (P_v)_i, \forall i \in \supp(v)] \\
        &= 3^{\wt(v)} \cdot \ab\Big(\frac{1}{3})^{\wt(v)} \cdot \tr(P_v \rhoout)\\
        &= \tr(P_v \rhoout).
    \end{align}
    \end{subequations}

    Now consider $\hat{C}_{\mathrm{in}}(w) = 3^{\wt(w)}\tr\ab\big(P_w \rhoin)$ for some Pauli operator $P_w$. Since $\rhoin = \bigotimes_{i=1}^n \rho_i$ with each $\rho_i$ independently and uniformly sampled, and $\tr(P_w \rhoin) = \prod_{i=1}^n \tr((P_w)_i \rho_i)$, we have
    \[
    \E_{\rhoin}\ab[3^{\wt(w)}\tr(P_w \rhoin)\; \rhoin] = \bigotimes_{i=1}^n \E_{\rho_i}\ab[3^{\mathbf{1}{[(P_w)_i\neq I]}} \tr\ab((P_w)_i \rho_i) \rho_i].
    \]
    Then it suffices to evaluate the expectation for each single qubit:
    \begin{enumerate}
        \item If $(P_w)_i = I$, then $3^{\mathbf{1}{[(P_w)_i\neq I]}} \tr\ab((P_w)_i \rho_i) \rho_i = 3^0 \tr(I \rho_i) \rho_i = \rho_i$. Taking the expectation gives $\E[\rho_i] = I/2$ since each $\rho_i$ is chosen independently and uniformly from the six Pauli eigenstates.
        \item If $(P_w)_i \in \{X, Y, Z\}$, among the six states only the two eigenstates of $(P_w)_i$ (denoted by $\ket{s_+}$ and $\ket{s_-}$) have non-zero trace with $(P_w)_i$. Therefore,
        \begin{align*}
        \E_{\rho_i}\ab[3^{\mathbf{1}{[(P_w)_i\neq I]}} \tr((P_w)_i \rho_i) \rho_i] &= 3 \E_{\rho_i}\ab[\tr((P_w)_i \rho_i) \rho_i]\\
        &= 3\ab\big( \frac{1}{6} \ketbra{s_+}{s_+} - \frac{1}{6} \ketbra{s_-}{s_-})\\
        &= \frac{1}{2}(\ketbra{s_+}{s_+} - \ketbra{s_-}{s_-})\\
        &= \frac{(P_w)_i}{2}.
        \end{align*}
    \end{enumerate}
    In both cases, the expectation on each single-qubit is $(P_w)_i/2$. Taking the tensor product gives
    \begin{equation}\label{eq:C_in}
        \E_{\rhoin}\ab[3^{\wt(w)}\tr(P_w \rhoin)\; \rhoin] = \bigotimes_{i=1}^n \frac{(P_w)_i}{2} = \frac{P_w}{2^n}.
    \end{equation}
    Combining \cref{eq:C_in,eq:C_out} yields
    \begin{align*}
        \E[\hat{R}_{vw}(t, m)] &= \E_{\rhoin}\ab\Big[\hat{C}_\mathrm{in}(w) \cdot \E_{B, o}\ab[\hat{C}_\mathrm{out}(v) \mathrel{|} \rhoin]]\\
        &= \E_{\rhoin}\ab\Big[3^{\wt(w)} \cdot \tr(P_w \rhoin) \cdot \tr(P_v \rhoout)]\\
        &= \E_{\rhoin}\ab\Big[3^{\wt(w)} \cdot \tr(P_w \rhoin) \cdot \tr\ab\big(P_v e^{t\cL}(\rhoin))]\\
        &= \E_{\rhoin}\ab\Big[  \tr\ab\Big(P_v e^{t\cL}\ab\big(3^{\wt(w)} \cdot \tr(P_w \rhoin)\cdot \rhoin))]\\
        &= \tr\ab(P_v e^{t\cL}\ab\Big(\E_{\rhoin}\ab[3^{\wt(w)} \cdot \tr(P_w \rhoin)\cdot \rhoin]))\\
        &= \tr\ab(P_v e^{t\cL}\ab\Big(\frac{P_w}{2^n}))\\
        &= \frac{1}{2^n} \tr\ab\big(P_v e^{t \cL}(P_w)) = R_{vw}(t). \qedhere
    \end{align*}
\end{proof}

\begin{lemma}[Variance of the PTM estimator]\label{lem:variance}
For every pair $(v, w) \in \cT_k$, $t \geq 0$, and $m \in [N]$
\[
  \E\ab\big[\hat R_{vw}(t, m)^2]=3^{\wt(v)+\wt(w)}, \qquad \Var\ab\big[\hat R_{vw}(t, m)] \leq 3^{\wt(v)+\wt(w)}.
\]
\end{lemma}

\begin{proof}
By definition $\hat R_{vw}(t, m) = \hat{C}_{\mathrm{in}}(w) \cdot\hat{C}_{\mathrm{out}}(v)$. Conditioned on $\rho_{\mathrm{in}}$, the output squares to
\[
  \hat{C}_{\mathrm{out}}(v)^2=\prod_{i\in\supp(v)}9\cdot\mathbf{1}[B_i=(P_v)_i]\,o_i^2 .
\]
Since $o_i^2=1$ and each measurement basis matches the required Pauli with probability $1/3$,
\[
  \E[\hat{C}_{\mathrm{out}}(v)^2\mid\rho_{\mathrm{in}}]=\prod_{i\in\supp(v)}9\cdot\tfrac13=3^{\wt(v)},
\]
independently of $\rho_{\mathrm{in}}$.  For the input, the six single-qubit Pauli eigenstates give $\E_\rho[(\tr Q\rho)^2]=1/3$ for each $Q\in\{X,Y,Z\}$, so
\[
  \E[\hat{C}_{\mathrm{in}}(w)^2]
  =3^{2\wt(w)}\prod_{i\in\supp(w)}\tfrac13
  =3^{\wt(w)}.
\]
By the tower property,
\[
  \E\ab\big[\hat R_{vw}(t, m)^2]
  =\E\ab[\hat{C}_{\mathrm{in}}(w)^2\,\E[\hat{C}_{\mathrm{out}}(v)^2\mid\rho_{\mathrm{in}}]]
  =3^{\wt(v)} \E\ab\big[\hat{C}_{\mathrm{in}}(w)^2]
  =3^{\wt(v)+\wt(w)}.
\]
The variance is at most the second moment, giving the stated bound.
\end{proof}

\begin{corollary}[Bias of the PTM generator estimator]\label{cor:estimator_of_L}
    Let $\cL$ be a $k$-local Lindbladian with local dynamical strength $\Lloc \leq \Lambda$, and set $T = \frac{1}{2k\Lambda}$, $q = \lceil\log(2k\Lambda/\eps)\rceil$. Then the estimator of the PTM generator $\hat{L}_{vw}$ in \cref{eq:shadow_deriv} satisfies that
    \[
    \abs[\big]{\E[\hat{L}_{vw}] - L_{vw}} \leq \eps/2
    \]
    for each pair $(v, w) \in \cT_k$.
\end{corollary}
\begin{proof}
    By definition, we have
    \[
    \hat{L}_{vw} = \ell'_0(0)\,\mathbf{1}[v=w] + \sum_{j=1}^{q}\ell'_j(0) \bar{R}_{vw}(t_j) = \ell'_0(0)\,\mathbf{1}[v=w] + \sum_{j=1}^{q}\ell'_j(0) \ab\Big(\frac{1}{N}\sum_{m=1}^{N}\hat{R}_{vw}(t_j, m)).
    \]
    By \cref{lem:unbiasedness_of_R}, the expectation is
    \begin{align*}
        \E[\hat{L}_{vw}] &= \ell'_0(0)\,\mathbf{1}[v=w] + \sum_{j=1}^{q}\ell'_j(0) \ab\Big(\frac{1}{N}\sum_{m=1}^{N} \E[\hat{R}_{vw}(t_j, m)]) \\
        &=\ell'_0(0)\,\mathbf{1}[v=w]+ \sum_{j=1}^{q}\ell'_j(0) R_{vw}(t_j).
    \end{align*}
    By \Cref{lem:truncation}, the Chebyshev derivative estimation with $T = \frac{1}{2k\Lambda}$ and $q = \lceil\log(2k\Lambda/\eps)\rceil$ yields
    \[
    \abs[\big]{\E[\hat{L}_{vw}] - L_{vw}} \leq (k\Lambda)^{q+1}T^q \leq \frac{\eps}{2}. \qedhere
    \]
\end{proof}

\subsection{Complexity for learning PTM generator}\label{sec:upper_bound}
To prove the upper bound for the sample complexity, we will need a concentration bound derived from weighted Bernstein's inequality.

\begin{lemma}[Weighted Bernstein's inequality {\cite[Theorem 2.10]{boucheron2013concentration}}]\label{lem:weighted_bernstein}
Let $\{X_{i}:1\leq i\leq m\}$ be independent random variables, and let $a_1,\ldots,a_m \in \bbR$ be real weights with $a_{\max}=\max_i\abs{a_i}$. Define the weighted sum and variance
\[
S \coloneq \sum_{i= 1}^m a_i (X_i - \E[X_i]), \qquad V \coloneq \sum_{i=1}^m a_i^2 \Var[X_i].
\]
Assume that $\abs{X_i - \E[X_i]} \leq B$ almost surely for each $i$. Then
\[
\Pr\ab[\abs{S}\geq s] \leq 2\exp\ab\Big( -\frac{s^2}{2(V+ a_{\max}sB/ 3)}).
\]
\end{lemma}

Next, we establish the sample complexity of the proposed algorithm. 
\begin{theorem}[Upper bound for PTM generator learning]\label{thm:ptm_generator_estimation}
    Let $\cL$ be an unknown $k$-local Lindbladian on $n$ qubits with local dynamical strength $\Lloc \leq \Lambda$. For any $\eps,\delta \in (0,1)$, the estimates $\hat{L}_{vw}$ produced by \cref{alg:PTM_learning} satisfy that
    \[
    \Pr\ab\Big[\max_{(v, w) \in \cT_k} \abs[\big]{\hat{L}_{vw}-L_{vw}}\geq \eps] \leq \delta,
    \]
    using a total number of 
    \[
    \widetilde O\ab\Big(\frac{\Lambda^2 }{\eps^2} 3^{2k} k^2  \rbra[\big]{k\log n+\log\tfrac{2}{\delta}})
    \]
    experiments. The total evolution time is
    \[
    \widetilde O\ab\Big(\frac{\Lambda}{\eps^2}3^{2k} k \rbra[\big]{k\log n+\log\tfrac{2}{\delta}}).
    \]
    Moreover, the classical processing of \cref{alg:PTM_learning} runs in time
    \[
    \widetilde O\ab\Big(\frac{ \Lambda^2}{\eps^2}\,3^{2k} k^3 n^k \rbra[\big]{k\log n+\log\tfrac{2}{\delta}}).
    \]
\end{theorem}
\begin{proof}
    In \cref{alg:PTM_learning}, we can see that
    \[
    \hat{L}_{vw} = \ell'_0(0)\,\mathbf{1}[v=w] + \sum_{j=1}^{q}\ell'_j(0) \bar{R}_{vw}(t_j) = \ell'_0(0)\,\mathbf{1}[v=w] + \sum_{j=1}^{q}\ell'_j(0) \ab\Big(\frac{1}{N}\sum_{m=1}^{N}\hat{R}_{vw}(t_j, m))
    \]
    is a weighted sum of $qN$ independent random variables $\hat{R}_{vw}(t_j, m)$, which have means $R_{vw}(t_j)$ and weights $\ell'_j(0)/N$.  The largest weight is $\ell'_{\max} / N$ with $\ell'_{\max} \coloneq \max_{j} \abs{\ell'_j(0)}$. Since $\abs{R_{vw}(t_j)}\leq 1$ and $\abs{\hat{R}_{vw}(t_j, m)}\leq3^{2k}$, we have $\abs{\hat{R}_{vw}(t_j, m) - R_{vw}(t_j)}\leq B = 2\cdot 3^{2k}$, and by \cref{lem:variance} the variance is $\Var[\hat{R}_{vw}(t_j, m)] \leq 3^{2k}$. The total weighted variance is therefore
    \begin{equation}\label{eq:weighted_var}
    V=\sum_{j=1}^q\sum_{m=1}^N\ab\Big(\frac{\ell'_j(0)}{N})^2\Var\ab\big[\hat{R}_{vw}(t_j, m)]\leq \frac{3^{2k}}{N}\sum_{j=1}^q \abs{\ell'_j(0)}^2.
    \end{equation}

    For Chebyshev--Lobatto nodes mapped to $[0, T]$, the chain rule introduces a factor of $2/T$ relative to the Chebyshev spectral differentiation matrix $D$ evaluated at the boundary. 
    We have
    \begin{equation*}
        \sum_{j=1}^q \ab\big|\ell'_j(0)| = \frac{2}{T} \sum_{j=1}^q \ab\big|D_{0j}|.
    \end{equation*}
    By \Cref{lem:chebyshev_matrix}, we have
    \begin{align*}
        \sum_{j=1}^q |D_{0j}| &= \ab\bigg(\sum_{j=1}^{q-1} \frac{2}{1 - \cos(j\pi/q)}) + \frac{2}{2(1 - \cos(\pi))} \\
        & = \sum_{j=1}^{q-1} \csc^2\ab\Big(\frac{j\pi}{2q}) + \frac{1}{2} \\
        & = \frac{2(q^2 - 1)}{3} + \frac{1}{2} = \frac{4q^2 - 1}{6}.
    \end{align*}
    Therefore, the weights sum to
    \begin{equation}
        \sum_{j=1}^q \ab\big|\ell'_j(0)| = \frac{2}{T}\ab\Big(\frac{4q^2 - 1}{6}) = O(k\Lambda q^2). \label{eq:lagrange_polynomial}
    \end{equation}
    By weighted Bernstein's inequality in \cref{lem:weighted_bernstein}, we have
    \begin{subequations}
        \begin{align}
            \Pr \ab[\abs*{\hat{L}_{vw}-\E[\hat{L}_{vw}]} > \frac{\eps}{2}]
            &\leq 2\exp\ab\bigg(
            -\frac{\eps^2}
            {8(V+\frac{\eps \ell'_{\max}B}{6N})})\\
            &\leq 2\exp\ab\bigg(
            -\frac{\eps^2}
            {8(\frac{3^{2k}}{N}\sum_{j=1}^q \abs{\ell'_j(0)}^2+ \frac{\eps \ell'_{\max}3^{2k}}{3N})}) \label{eq:var_bound}\\
            &= 2\exp\ab\bigg(
            -\frac{N \eps^2}
            {8\cdot {3^{2k}} (\sum_{j=1}^q \abs{\ell'_j(0)}^2+ \eps \ell'_{\max}/{3})})\\
            &\leq 2\exp\ab\bigg(
            -\frac{N \eps^2}
            {8\cdot {3^{2k}} (\sum_{j=1}^q \abs{\ell'_j(0)}^2+ (\sum_{j=1}^q \abs{\ell'_j(0)})^2/{3})}) \label{eq:l_bound}\\
            &\leq 2\exp\ab\bigg(
            -\frac{N \eps^2}
            {32\cdot {3^{2k-1}}  (\sum_{j=1}^q \abs{\ell'_j(0)})^2}) \label{eq:lsum_bound}\\
            &\leq 2\exp\ab\bigg(
            -\frac{N \eps^2}
            {32\cdot {3^{2k-1}} \cdot O(k^2 \Lambda^2 q^4)}), \label{eq:lag_bound}
        \end{align}
    \end{subequations}
    where \cref{eq:var_bound} follows from \cref{eq:weighted_var} and $B = 2 \cdot 3^{2k}$, \cref{eq:l_bound} holds because $\eps,\ell'_{\max} \leq \sum\nolimits_{j=1}^q \abs{\ell'_j(0)}$, \cref{eq:lsum_bound} relies on $\sum\nolimits_{j=1}^q \abs{\ell'_j(0)}^2 \leq (\sum\nolimits_{j=1}^q \abs{\ell'_j(0)})^2$, and \cref{eq:lag_bound} follows from \cref{eq:lagrange_polynomial}.
    Recall that the number of pairs in $\cT_k$ is $M_k = O(n^{k})$. Applying the union bound over all $M_k$ pairs and absorbing the $\eps/2$ truncation error from \cref{cor:estimator_of_L}, we have
    \[
    \Pr\ab\Big[\max_{(v, w) \in \cT_k} \abs[\big]{\hat{L}_{vw}-L_{vw}}\geq \eps] \leq 2 M_k \exp\ab\Big(-\frac{N\eps^2}{O(k^2 \Lambda^2 q^4\cdot 3^{2k})}).
    \]
    Bounding the probability by $\delta$ gives
    \[
    2M_k \exp\ab\Big(-\frac{N\eps^2}{O(k^2 \Lambda^2 q^4\cdot 3^{2k})}) \leq \delta
    \implies
    N = O\ab\Big(\frac{3^{2k} k^2 \Lambda^2 q^4 }{\eps^2}\rbra[\big]{k\log n+\log\tfrac{2}{\delta}}).
    \]
    Since $q = \widetilde{\Theta}(1)$ by \Cref{lem:truncation}, the experiment count in \cref{alg:PTM_learning} is
    \[
    qN = \widetilde O\ab\Big(\frac{3^{2k} k^2 \Lambda^2 }{\eps^2}\rbra[\big]{k\log n+\log\tfrac{2}{\delta}}).
    \]
    The evolution time is at most $T = 1/(2k\Lambda)$ per experiment, therefore the total evolution time is
    \[
        \Ttot \leq qN \cdot T = \widetilde O\ab\Big(\frac{3^{2k} k \Lambda }{\eps^2}\rbra[\big]{k\log n+\log\tfrac{2}{\delta}}).
    \]

    Finally, we bound the classical running time. In Stage~1, the algorithm generates and stores $qN$ records, each consisting of an $n$-qubit product-state label and an $n$-qubit Pauli-basis/outcome pair, taking $O(nqN)$ time. In Stage~2, for each of the $M_k$ target entries $(v,w)$ and each of the $qN$ records, it evaluates the single-shot estimator $\hat R_{vw}(t_j,m)=\hat C_{\mathrm{in}}(w)\,\hat C_{\mathrm{out}}(v)$. Since $\hat C_{\mathrm{in}}(w)$ and $\hat C_{\mathrm{out}}(v)$ depend only on the at most $k$ qubits in $\supp(w)$ and $\supp(v)$, each evaluation costs $O(k)$, for a total of $O(k\,M_k\,qN)$. Forming the $q$ empirical averages and the endpoint-derivative combination \cref{eq:shadow_deriv} adds only $O(q\,M_k)$. The classical running time is therefore
    \[
      O\ab\big(k M_k qN + n qN) = \widetilde O\ab\Big(\frac{3^{2k}k^3 n^k \Lambda^2}{\eps^2}\rbra[\big]{k\log n+\log\tfrac{2}{\delta}}),
    \]
    using $M_k = O(n^{k})$, $M_k \geq n$, and $q=\widetilde\Theta(1)$.
\end{proof}

\subsection{Coefficient recovery from PTM generator entries}\label{sec:coeff_recovery}

This subsection converts the estimated PTM generator $\{\hat L_{vw}\}$ into the physical coefficients in three steps: (i) we relate the $\chi$-matrix entries to the coefficients (\cref{lem:chi_to_params}); (ii) we show the map $\{L_{vw}\}\to\{\chi_{bc}\}$ decouples by Pauli shift into local Walsh--Hadamard transforms as \cref{eq:fwht_inverse}; and (iii) we remove support aliasing by descending de-aliasing (\cref{alg:coeff_recovery_dense} and \cref{alg:coeff_recovery_threshold}), with the end-to-end error analysis in \cref{lem:recovery_error,thm:learning_upper}.

Every linear map $\cL$ on $n$-qubit operators has a unique $\chi$-matrix expansion
\begin{equation}\label{eq:chi}
  \cL(A)=\sum_{b,c}\chi_{bc}P_b A P_c.
\end{equation}
The PTM generator and the $\chi$-matrix are two different representations of the same linear map $\cL$.
The PTM generator entries
\[
  L_{vw}=\frac{1}{2^n}\tr(P_v\cL(P_w))
\]
describe how $\cL$ maps an input Pauli operator $P_w$ into the Pauli basis. In contrast, the $\chi$-matrix expansion expresses $\cL$ in the left--right Pauli multiplication basis. These two representations are not equal entrywise, but they are related by an invertible change of basis.

For the Lindbladian in \cref{eq:lindbladian}, the entries of this $\chi$ matrix determine the physical coefficients as follows.

\begin{lemma}[Coefficient recovery from $\chi$-matrix]\label{lem:chi_to_params}
For all Pauli operators $a\in \cP_{n,k}$ and all Pauli pairs $(b,c) \in \Gamma_k$, we have
\[
  h_a=\frac{\iu}{2}\ab(\chi_{a0}-\chi_{0a}),\qquad \gamma_{bc}=\chi_{bc}.
\]
\end{lemma}

\begin{proof}
The terms $P_b A P_c$ with $(b,c)\in\Gamma_k$ arise only from the first term of the dissipator in \cref{eq:lindbladian}, so $\gamma_{bc}=\chi_{bc}$. For the boundary entries, the Hamiltonian contributes $-\iu h_a$ to $\chi_{a0}$ and $+\iu h_a$ to $\chi_{0a}$.  The anti-commutator part contributes the same scalar to both entries:
\[
  -\frac{1}{2}\sum_{\substack{(b,c)\in\Gamma_k\\b\oplus c=a}}
      \gamma_{bc}\,\xi_{cb}.
\]
Thus $\chi_{a0}-\chi_{0a}=-2\iu h_a$, which is equivalent to the claimed formula.
\end{proof}

The recovery is organized by Pauli shifts.  Substituting \cref{eq:chi} into the PTM generator gives
\[
  L_{vw}=\sum_{b,c}\chi_{bc}\ab[2^{-n}\tr(P_v P_b P_w P_c)],
\]
and the trace is zero unless
\[
  v\oplus b \oplus w\oplus c=0.
\]
Thus, only terms with the same shift
\[
  u\coloneq v\oplus w=b\oplus c
\]
mix with one another. By expanding $L_{w \oplus u, w}$, we have
\begin{align*}
    L_{w \oplus u, w} &= \frac{1}{2^n} \sum_{b} \chi_{b, b\oplus u} \tr\ab(P_{w \oplus u} P_b P_{w} P_{b \oplus u})\\
    &= \frac{1}{2^n} \sum_{b} \chi_{b, b\oplus u} (-1)^{\sinprod{b}{w}} \tr\ab(P_{w \oplus u} P_w P_b P_{b \oplus u})\\
    &= \frac{1}{2^n} \sum_{b} \chi_{b, b\oplus u} (-1)^{\sinprod{b}{w}} \xi_{w \oplus u, w} \xi_{b,b \oplus u}  \tr\ab( P_u P_u)\\
    &= \sum_{b} \chi_{b, b\oplus u} (-1)^{\sinprod{b}{w}} \xi_{w \oplus u, w} \xi_{b,b \oplus u}.
\end{align*}
Grouping all known classical phases gives
\begin{equation}\label{eq:fwht_forward}
  \xi_{w\oplus u,w}^{-1}L_{w\oplus u,w}
  =\sum_{b}\xi_{b, b \oplus u}\chi_{b,b\oplus u}(-1)^{\sinprod{b}{w}}.
\end{equation}
For notational convenience, define
\[
  F_u(w)\coloneq \xi_{w\oplus u,w}^{-1} L_{w\oplus u,w},
  \qquad
  G_u(b)\coloneq \xi_{b,b\oplus u}\chi_{b,b\oplus u}.
\]
Then \cref{eq:fwht_forward} becomes
\[
F_u(w)=\sum_{b'}G_u(b')(-1)^{\sinprod{b'}{w}}.
\]
Fix a qubit set $S$ with $\abs{S}\leq k$ and a shift $u$ supported on $S$. Restrict to Pauli operators supported inside $S$, i.e., $\supp(w) \subseteq S$, and write the function as $F_u^S(w)$. Since $w$ is the identity outside $S$, the commutation sign only depends on the restriction of $b'$ to $S$:
\[
\sinprod{b'}{w} = \sinprod{b'|_S}{w}.
\]
Therefore,
\begin{equation}
  F_u^S(w) = \sum_{b'} G_u(b')(-1)^{\sinprod{b'|_S}{w}}.
\end{equation}
We can further group global coefficients by their local restriction and define the aggregated local quantity
\begin{equation}
    G_u^S(b) \coloneq \sum_{b': b'|_S = b} G_u(b'), \qquad \text{where } b \in \bits^{2S}.
\end{equation}
Regrouping the global sum gives
\begin{equation}
  F^S_u(w) = \sum_{b \in \bits^{2S}} G_u^S(b)(-1)^{\sinprod{b}{w}},
\end{equation}
which is precisely the Walsh--Hadamard transform on the local region $S$. In matrix notation, we can write
\[
F_u^S = H_S G_u^S,
\]
where $H_S$ is the Hadamard matrix with entries $(-1)^{\sinprod{b}{w}}$. The Hadamard matrix is proportional to an orthogonal matrix, satisfying
\[(H_S)^\top H_S=4^{\abs{S}}I,\]
so the inversion has condition number exactly $1$, and the PTM entry estimation error propagates to $\chi$-matrix entry error without amplification. The inverse Walsh--Hadamard transform is
\begin{equation}
  G_u^S(b) = \frac{1}{4^{\abs{S}}} \sum_{w \in \bits^{2S}} F^S_u(w)(-1)^{\sinprod{b}{w}}.
\end{equation}
Substituting back the definitions of $F_u^S$ and $G_u^S$ gives
\begin{equation}
  \sum_{b': b'|_S = b} \xi_{b',b'\oplus u}\chi_{b',b'\oplus u} = \frac{1}{4^{\abs{S}}} \sum_{w \in \bits^{2S}} \xi_{w\oplus u,w}^{-1} L_{w\oplus u,w} (-1)^{\sinprod{b}{w}}.
\end{equation}
Since $u$ is supported on $S$, the phase $\xi_{b',b'\oplus u}$ depends only on the restriction $b'|_S$, hence
\[
\xi_{b',b'\oplus u} = \xi_{b'|_S,b'|_S \oplus u} = \xi_{b, b \oplus u}.
\]
For a Pauli operator $b \in \bits^{2S}$, define the local extension sum
\[
\chi_{b, b\oplus u}^S \coloneq \sum_{b': b'|_S = b} \chi_{b',b'\oplus u}.
\]
It follows that
\begin{equation}\label{eq:fwht_inverse}
  \chi_{b, b\oplus u}^S = \frac{1}{\xi_{b, b \oplus u} 4^{\abs{S}}} \sum_{w \in \bits^{2S}} \xi_{w\oplus u,w}^{-1} L_{w\oplus u,w} (-1)^{\sinprod{b}{w}}.
\end{equation}
In other words, the inverse Walsh--Hadamard transform does not generally return a single $\chi$-matrix entry, but the aggregated sum $\chi^S_{b,b\oplus u}$ of all global entries whose restriction to $S$ equals the local label $b$.

The remaining step is to extract the individual entry by de-aliasing, i.e., removing common outside extensions without knowing the true support. Fix a shift $u$ and a candidate label $b$, where either $(b,b\oplus u)\in\Gamma_k$, or $u\in\cP_{n,k}$ and $b\in\{0,u\}$ for the two Hamiltonian boundary entries.  Let
\[
  S=\supp(b, b\oplus u).
\]
Then the local extension sum over $S$ obeys
\[
  \chi^{S}_{b,b\oplus u}=\chi_{b,b\oplus u}+
  \sum_{\substack{b'\neq b\\b'|_{S}=b|_{S}\\\gamma_{b',b'\oplus u}\neq0}}\chi_{b',b'\oplus u},
\]
with the convention that boundary entries $\chi_{u0}$ and $\chi_{0u}$ are also kept when they appear. The local extension sum equals the desired coefficient plus the nonzero larger dissipative coefficients that agree with $b$ on the local support $S$. Since these common outside extensions are unknown, the post-processing below peels them from larger supports to smaller supports. This estimates every dissipative coefficient indexed by $\Gamma_k$.  Exact identification of the nonzero support would require an additional coefficient gap assumption, which we do not impose.

We give two de-aliasing algorithms, both processing candidates from larger to smaller support. The dense algorithm (\cref{alg:coeff_recovery_dense}) makes no assumption on the support and subtracts every larger candidate alias. The thresholded peeling algorithm (\cref{alg:coeff_recovery_threshold}) instead uses the dissipative site-degree bound to subtract only the larger aliases retained above threshold. Our main theorem (\cref{thm:learning_upper}) runs whichever of the two has the smaller recovery factor.

For the stability bounds below, we define two recovery factors,
$\Kden$ and $\Kdeg$, which bound the errors of \cref{alg:coeff_recovery_dense}
and \cref{alg:coeff_recovery_threshold} respectively in \cref{lem:recovery_error}.
First define the dense recovery
factor.  For $1\leq s<r\leq k$, let
\[
  B_{s,r} \coloneq 3^{r-s}\binom{n-s}{r-s}.
\]
This is the number of ways to enlarge a fixed union support of size
$s$ to one of size $r$ by choosing the new qubits and assigning one
of the three non-identity Paulis to the common outside extension.  Define
\[
  F_k\coloneq 1,\qquad F_s \coloneq 1+\sum_{r=s+1}^k B_{s,r} F_r, \qquad s=k-1,\ldots,1,
\]
and set
\[
  \Kden\coloneq\max_{1\leq s\leq k}F_s.
\]
For fixed $k$, we have $\Kden= O(n^{k-1})$. The
bounded-degree thresholded peeling algorithm has recovery factor
\[
  \Kdeg \coloneq 2\sum_{j=0}^{k-1}(2\Ddeg)^j.
\]
For fixed $k = O(1)$ and bounded dissipative site degree $\Ddeg=O(1)$,
this is a constant. Finally set
\[
  \Krec \coloneq \min\{\Kdeg,\Kden \}.
\]

\begin{algorithm}[ht!]
\caption{Dense de-aliasing coefficient recovery}
\label{alg:coeff_recovery_dense}
\begin{algorithmic}[1]
\Require Estimates $\hat L_{vw}$ for $(v,w)\in\cT_k$; locality $k$; accuracy parameter $\eta$.
\Ensure Estimates of the Lindbladian coefficients $\{\hat{h}_a:a\in\cP_{n,k}\}$ and $\{\hat\gamma_{bc}:(b,c)\in\Gamma_k\}$.
\For{each shift $u$ with $\wt(u)\leq k$}
  \For{$s=k,k-1,\ldots,1$}
    \For{each $b$ with $\wt(b,b\oplus u)=s$ and either $(b,b\oplus u)\in\Gamma_k$ or $u\in\cP_{n,k}$ and $b\in\{0,u\}$}
      \State Set $S\gets\supp(b,b\oplus u)$ and compute $\hat\chi^S_{b,b\oplus u}$ by \cref{eq:fwht_inverse}.
      \State Set the full alias set
      \[
        \hat\cA_{u,b} \gets \cbra{b':(b',b'\oplus u)\in\Gamma_k,\ \wt(b',b'\oplus u)>s,\ b'|_S=b|_S}.
      \]
      \State Form the fully de-aliased residual
      \[
        Z_{b,u} \gets \hat\chi^{S}_{b,b\oplus u} - \sum_{b'\in\hat\cA_{u,b}} \hat\chi_{b',b'\oplus u}.
      \]
      \State Set $\hat\chi_{b,b\oplus u}\gets Z_{b,u}$.
    \EndFor
  \EndFor
\EndFor
\State\Return $\hat h_a \gets \frac{\iu}{2}(\hat\chi_{a0}-\hat\chi_{0a})$ for $a\in\cP_{n,k}$ and $\hat\gamma_{bc}\gets\hat\chi_{bc}$ for $(b,c)\in\Gamma_k$.
\end{algorithmic}
\end{algorithm}

\begin{algorithm}[ht!]
\caption{Thresholded peeling coefficient recovery with low-intersection dissipators}
\label{alg:coeff_recovery_threshold}
\begin{algorithmic}[1]
\Require Estimates $\hat L_{vw}$ for $(v,w)\in \cT_k$; locality $k$; dissipative site degree $\Ddeg$; accuracy parameter $\eta$.
\Ensure Estimates of the Lindbladian coefficients $\{\hat{h}_a:a\in\cP_{n,k}\}$ and $\{\hat\gamma_{bc}:(b,c)\in\Gamma_k\}$.
\State Set $E_{k+1}\gets0$.
\For{$s=k,k-1,\ldots,1$}
  \State Set $\tau_s\gets \eta+\Ddeg E_{s+1}$ and $E_s\gets2\tau_s$.
\EndFor
\For{each shift $u$ with $\wt(u)\leq k$}
  \State Initialize the dissipative set $\hat\cS_u\gets\varnothing$.
  \For{$s=k,k-1,\ldots,1$}
    \For{each $b$ with $\wt(b,b\oplus u)=s$ and either $(b,b\oplus u)\in\Gamma_k$ or $u\in\cP_{n,k}$ and $b\in\{0,u\}$}
      \State Set $S \gets \supp(b,b\oplus u)$ and compute $\hat\chi^S_{b,b\oplus u}$ by \cref{eq:fwht_inverse}.
      \State Form the peeled residual
      \[
      Z_{b,u} \gets \hat\chi^{S}_{b,b\oplus u} - \sum_{\substack{b'\in\hat\cS_u:\\b'|_{S}=b|_{S}}} \hat\chi_{b',b'\oplus u}.
      \]
      \If{$u \in \cP_{n,k}$ and $b \in \{0,u\}$}
        \State Set $\hat\chi_{b,b\oplus u} \gets Z_{b,u}$.
      \ElsIf{$\abs{Z_{b,u}}>\tau_s$}
        \State Set $\hat\chi_{b,b\oplus u} \gets Z_{b,u}$.
        \State Add $b$ to $\hat\cS_u$.
      \Else
        \State Set $\hat\chi_{b,b\oplus u} \gets 0$.
      \EndIf
    \EndFor
  \EndFor
\EndFor
\State\Return $\hat h_a \gets \frac{\iu}{2}(\hat\chi_{a0}-\hat\chi_{0a})$ for $a \in \cP_{n,k}$ and $\hat\gamma_{bc} \gets \hat\chi_{bc}$ for $(b,c)\in\Gamma_k$. 
\end{algorithmic}
\end{algorithm}

\begin{lemma}[Coefficient recovery stability]\label{lem:recovery_error}
Assume that
\[
  \max_{(v,w)\in \cT_k}
  \abs{\hat L_{vw}-L_{vw}}\leq \eta.
\]
Run either \cref{alg:coeff_recovery_dense} or \cref{alg:coeff_recovery_threshold}. Then every local extension sum computed by the algorithm has error at most $\eta$.  Moreover, every candidate dissipative entry $(b,c)\in\Gamma_k$ and every boundary entry used for Hamiltonian recovery satisfies 
\[
  \abs{\hat\chi_{bc}-\chi_{bc}} \leq \Kden\,\eta
  \ \text{ (\cref{alg:coeff_recovery_dense})}, \qquad
  \abs{\hat\chi_{bc}-\chi_{bc}} \leq \Kdeg\,\eta
  \ \text{ (\cref{alg:coeff_recovery_threshold})},
\]
where boundary entries are indexed by the same union support size $\wt(b,b\oplus u)$. In particular, running the algorithm with the smaller recovery factor yields error at most $\Krec\,\eta$ with $\Krec=\min\{\Kdeg,\Kden\}$.
\end{lemma}

\begin{proof}
For a fixed local inversion region $S$, all phases in \cref{eq:fwht_inverse} have modulus one.  Subtracting the exact inverse from \cref{eq:fwht_inverse} gives
\[
  \abs[\big]{\hat\chi^S_{b,b\oplus u}-\chi^S_{b,b\oplus u}} \leq 4^{-\abs{S}}\sum_{w \in \bits^{2S}} \abs{\hat L_{w\oplus u,w}-L_{w\oplus u,w}} \leq \eta.
\]
First consider \cref{alg:coeff_recovery_dense}.  Let $E_s$ be the maximum
recovery error among all dissipative candidates and boundary
entries with union support size $s$. For $s=k$, there are no strict larger extensions, so $E_k\leq\eta=F_k\eta$. If
$\abs{S}=s<k$, a strict common extension to union support size $r$ is obtained by choosing $r-s$ qubits outside $S$ and placing the same non-identity Pauli on each new qubit in the two labels. Hence the number of such extensions is at most $B_{s,r}$. Since dense de-aliasing subtracts all larger candidates, the only remaining error is the local-inversion error plus the accumulated errors of the already reconstructed larger entries,
\[
  E_s \leq \eta + \sum_{r=s+1}^k B_{s,r}E_r.
\]
The recursive definition of $F_s$ gives $E_s\leq F_s\eta$ by descending induction, and thus $E_s\leq \Kden\eta$.

Now consider \cref{alg:coeff_recovery_threshold}, and reuse the symbol $E_s$, now redefined as the maximum recovery error at union support size $s$ for this algorithm (governed by the thresholds $\tau_s$ rather than the dense recursion $F_s$ above). We call a candidate $b$ \emph{active} if its true coefficient $\chi_{b,b\oplus u}$ is nonzero, and \emph{inactive} otherwise. Fix a shift $u$ and a candidate $b$, set $S=\supp(b,b\oplus u)$, and write $r=\abs{S}$.  Any strict common extension $b'\neq b$ has $b'|_S=b|_S$ and $\supp(b',b'\oplus u)\supsetneq S$. Fix any $i\in S=\supp(b,b\oplus u)$. Since $b'|_S=b|_S$ and $u$ is supported on $S$, we have $(b'\oplus u)|_S=(b\oplus u)|_S$, so the pair $(b',b'\oplus u)$ acts nontrivially on $i$ exactly as $(b,b\oplus u)$ does. By the definition of dissipative site degree, at most $\Ddeg$ dissipative pairs act on $i$, so there are at most $\Ddeg$ true larger aliases.

We prove the claim by descending induction on $r$.  Assume that all coefficients with support size larger than $r$ have already been estimated to error at most $E_{r+1}$, and that no inactive larger candidate has been retained.  For the residual $Z_{b,u}$, the local extension sum error contributes at most $\eta$.  The only additional residual error comes from at most $\Ddeg$ true larger aliases, each contributing at most $E_{r+1}$, whether it was retained and subtracted with error or was too small to be retained. Hence
\[
  \abs{Z_{b,u}-\chi_{b,b\oplus u}} \leq \eta+\Ddeg E_{r+1} =\tau_r.
\]
The boundary labels are the only entries needed to recover Hamiltonian coefficients, and their strict common extensions are dissipative pairs, so the same $\Ddeg$ bound controls their aliasing.  If $b$ is a boundary label, the algorithm keeps $Z_{b,u}$ and its error is at most $\tau_r\leq E_r$.  If $b$ is dissipative and $\abs{Z_{b,u}}>\tau_r$, the algorithm keeps $Z_{b,u}$ and again the error is at most $\tau_r\leq E_r$.  If $\abs{Z_{b,u}}\leq\tau_r$, the algorithm outputs zero, then
\[
  \abs{\chi_{b,b\oplus u}} \leq \abs{Z_{b,u}}+\tau_r \leq 2\tau_r = E_r.
\]
For an inactive dissipative candidate, the same bound gives $\abs{Z_{b,u}}\leq \tau_r$, so it is not retained. This closes the induction. Unrolling $E_{k+1}=0$ and $E_r=2(\eta+\Ddeg E_{r+1})$ gives
\[
  E_r = 2\eta\sum_{j=0}^{k-r}(2\Ddeg)^j \leq \Kdeg\,\eta,
\]
which proves the claimed coefficient bound for \cref{alg:coeff_recovery_threshold}.
\end{proof}

We also bound the classical running time of the recovery procedure.
\begin{lemma}[Running time of coefficient recovery]\label{lem:recovery_time}
Given the PTM-generator estimates $\{\hat L_{vw}\}_{(v,w)\in\cT_k}$, \cref{alg:coeff_recovery_dense} runs in classical time
$\widetilde O(3^{O(k)}n^k)$ and
\cref{alg:coeff_recovery_threshold} runs in classical time $\widetilde O(3^{O(k)} (1+\Ddeg) n^k)$.
\end{lemma}
\begin{proof}
Let $C_k\coloneq\abs{\Gamma_k}+\abs{\cP_{n,k}}=3^{O(k)}n^k$ be the
number of recovered coefficients.  Both recovery algorithms enumerate,
for each shift $u$ with $\wt(u)\leq k$, the candidate labels $b$. The
total number of such $(u,b)$ candidates is $O(C_k)$.  For a fixed shift
$u$ and region $S=\supp(b,b\oplus u)$ with $\abs{S}\leq k$, the inverse
Walsh--Hadamard transform \cref{eq:fwht_inverse} returns all local
extension sums $\{\hat\chi^S_{b,b\oplus u}\}_{b|_S}$ simultaneously by a
fast Walsh--Hadamard transform on $4^{\abs{S}}\leq4^k$ points in
$O(4^k k)$ time.  Summed over all relevant regions and shifts, these
local transforms cost $\widetilde O(16^k n^k)=\widetilde O(3^{O(k)} n^k)$.

It remains to count the cost of residual subtraction. For the dense de-aliasing algorithm, implement the alias set by enumerating larger
extensions directly, rather than by scanning all of $\Gamma_k$.  Fix a candidate with union support size $s$. The number of its larger aliases with union support size $r>s$ is at most
$B_{s,r}=3^{r-s}\binom{n-s}{r-s}$. The number of candidates with union support size $s$ is at most $3^{O(s)}\binom ns$. Hence the total number of dense alias subtractions is bounded by
\[
  \sum_{1\leq s<r\leq k} 3^{O(s)}\binom ns\, 3^{r-s}\binom{n-s}{r-s} = \sum_{1\leq s<r\leq k} 3^{O(k)}\binom{n}{r}\binom{r}{s} = 3^{O(k)}n^k .
\]
Thus dense de-aliasing has total classical running time
$\widetilde O(3^{O(k)}n^k)$.

For the thresholded de-aliasing algorithm, precomputing the thresholds $\{\tau_s,E_s\}_{s=1}^k$ for
\cref{alg:coeff_recovery_threshold} takes $O(k)$ time.  The thresholded
peeling forms each residual $Z_{b,u}$ by subtracting at most $\Ddeg$
retained larger aliases (by the proof of \cref{lem:recovery_error},
only active candidates are ever retained),
located in $O(\Ddeg)$ time with a hash table keyed by $(u,b|_S)$.
To make every retained label findable by the smaller candidates it aliases,
we insert it under all of its $\le 2^k$ sub-region keys at retention time,
which adds $O(2^k C_k)=3^{O(k)}n^k$ in total. The lookups themselves
cost $O(\Ddeg\,C_k)=3^{O(k)}\Ddeg\,n^k$ in total.  Combining this with the transform cost
gives $\widetilde O(3^{O(k)} (1+\Ddeg) n^k)$.
\end{proof}

Combining the PTM-generator estimation of \cref{thm:ptm_generator_estimation} with the coefficient recovery of \cref{lem:recovery_error,lem:recovery_time}, we obtain our main upper bound for local Lindbladian learning.

\begin{theorem}[Upper bound for local Lindbladian learning]\label{thm:learning_upper}
Let $\cL$ be an unknown $k$-local Lindbladian on $n$ qubits with dissipative site degree $\Ddeg$ and local dynamical strength $\Lloc \leq \Lambda$. For any $\eps, \delta \in (0,1)$, there exists an algorithm that outputs estimates of all coefficients of $\cL$ such that with probability at least $1-\delta$,
\[
  \max_a\abs{\hat h_a-h_a}\leq \eps, \qquad \max_{(b,c)\in\Gamma_k}\abs{\hat\gamma_{bc}-\gamma_{bc}}\leq \eps.
\]
The experiment count is
\[
  \widetilde O\ab(3^{O(k)} \Krec^{2}\frac{\Lambda^2}{\eps^2} \ab\Big(k\log n+\log\frac1\delta)),
\]
and the total evolution time is 
\[
  \widetilde O\ab(3^{O(k)} \Krec^{2} \frac{\Lambda}{\eps^2} \ab\Big(k\log n+\log\frac1\delta)).
\]
Moreover, the algorithm uses classical processing time
\[
  \widetilde O\ab(3^{O(k)} \Krec^{2}\frac{\Lambda^2}{\eps^2}\,n^{k} \ab\Big(k\log n+\log\frac1\delta)).
\]
In particular, $\Krec = O(1)$ when $\Ddeg = O(1)$, and $\Krec = O(n^{k-1})$ unconditionally.
\end{theorem}

\begin{proof}
Run \cref{alg:PTM_learning} with PTM accuracy parameter $\eta=\eps/\Krec$, and then run \cref{alg:coeff_recovery_dense} when
$\Kden\leq\Kdeg$, and \cref{alg:coeff_recovery_threshold}
otherwise.

Every local Fourier inversion is over $S=\supp(b,b\oplus u)$ with $\abs{S}\leq k$, so its PTM entries have $\supp(v)\cup\supp(w)\subseteq S$, hence $\wt(v,w)\leq k$, and therefore lie in $\cT_k$. By \cref{thm:ptm_generator_estimation}, all these PTM entries are estimated to error at most $\eta$ simultaneously with probability at least $1-\delta$.  By \cref{lem:recovery_error}, every recovered individual $\chi$ entry used by \cref{alg:coeff_recovery_dense} or \cref{alg:coeff_recovery_threshold} has error at most $\Krec\eta=\eps$.

The formula for $\gamma$ is direct from \cref{lem:chi_to_params}.  For $h_a$,
\[
  \abs{\hat h_a-h_a} \leq \frac{1}{2}\abs{\hat\chi_{a0}-\chi_{a0}}
  + \frac{1}{2} \abs{\hat\chi_{0a}-\chi_{0a}} \leq \eps.
\]
The bound on the experiment count follows from \cref{thm:ptm_generator_estimation} with $\eta=\eps/\Krec$ in place of $\eps$, and similarly for the total evolution time.

For the classical running time, \cref{thm:ptm_generator_estimation} with $\eta=\eps/\Krec$ shows that \cref{alg:PTM_learning} runs in time $\widetilde O(3^{O(k)}\Krec^2\Lambda^2 n^{k}/\eps^2\,(k\log n+\log\tfrac1\delta))$. By \cref{lem:recovery_time}, \cref{alg:coeff_recovery_dense} or \cref{alg:coeff_recovery_threshold} adds only $\widetilde O(3^{O(k)} (1+\Ddeg) n^k)$, which is dominated by the sampling stage. Summing the two stages gives the stated classical processing time.
\end{proof}

\section{Lower bounds for Lindbladian learning}\label{sec:lowerbound}
We prove information-theoretic lower bounds matching \cref{thm:learning_upper,thm:ptm_generator_estimation} up to logarithmic factors. The bounds hold in the strongest access model, allowing adaptive learners to use arbitrary ancillas, entangling operations, and arbitrary measurements, certifying that our non-adaptive, ancilla-free algorithm is near-optimal and that neither adaptivity nor entanglement helps. \cref{sec:lower_bound_techniques} sets up the hard instance and the limitation on information gained per experiment. \cref{thm:lower_bound_measurement} gives the lower bound on the experiment count, and \cref{thm:lower_bound_time} gives the evolution-time bound.

The lower-bound strategy has two steps.  First, choose two similar Lindbladians whose coefficients differ by more than the target accuracy. Any successful learner must distinguish them.  Second, bound the relative entropy between the learner's final output distributions under the two hypotheses. The bound holds even if the learner uses ancillas, adaptive input states, and adaptive measurements.

\subsection{Construction of hard instances}\label{sec:lower_bound_techniques}

First, we construct the hard instance. For $\gamma \geq 0$, define the one-qubit dephasing Lindbladian
\begin{equation}
    \cL_\gamma(\rho) = \gamma(Z \rho Z - \rho),
\end{equation}
whose locality is $k=1$ and dissipative site degree is  $\Ddeg = 1$. For an observable $Q$, the Heisenberg-picture action can be written as $\cL_\gamma^\dagger(Q) = \gamma(ZQZ-Q)$, which gives $\cL_\gamma^\dagger(Q)=0$ when $[Z,Q]=0$ and $\cL_\gamma^\dagger(Q)=-2\gamma Q$ when $\{Z,Q\}=0$. Therefore, by the definition of induced operator norm, we have
\begin{equation}
    \norm{\cL_\gamma^\dagger}_{\infty\to \infty} = 2\gamma.
\end{equation}
Since this Lindbladian consists of a single local term on one qubit, its local dynamical strength coincides with the operator norm, $\Lloc = \norm{\cL_\gamma^\dagger}_{\infty\to\infty} = 2\gamma$.

Given $\Lambda > 0$ and $\eps \in (0, \Lambda/16]$, we consider two instances $\cL_{\gamma_0}$ and $\cL_{\gamma_1}$ by setting
    \[
    \gamma_0 \coloneq \frac{\Lambda}{4} \qquad \text{and} \qquad \gamma_1 \coloneq \gamma_0 + 4 \eps.
    \]
    Since $\eps \leq \Lambda/16$, we have $2\gamma_1 = \Lambda/2+8\eps \leq \Lambda$, so both instances clearly satisfy the local-strength constraint $\Lloc \leq \Lambda$. Note that the coefficients of $\cL_{\gamma_0}$ and $\cL_{\gamma_1}$ differ in $\ell^\infty$ norm by $\abs{\gamma_0 - \gamma_1} = 4\eps$, so any algorithm that learns the Lindbladian to $\ell^\infty$ error $\eps$ can distinguish $\gamma_0$ from $\gamma_1$.

The channel semigroup $e^{t\cL_\gamma}$ is a family of Pauli dephasing channels
\[
e^{t\cL_\gamma}(\rho) = (1-p_{\gamma t})\rho + p_{\gamma t} Z\rho Z, \qquad p_{\gamma t} \coloneq \frac{1- e^{-2\gamma t}}{2}.
\]
Thus, an experiment with evolution time $t$ can be viewed as applying a Pauli $Z$ conjugation with probability $p_{\gamma t}$ governed by a Bernoulli distribution.
Revealing this hidden Bernoulli flag can only increase the information available to the learner, so the relative entropy between two dephasing channels is bounded by the relative entropy between their corresponding Bernoulli distributions.

\subsection{Experiment count lower bound}

We first prove the lower bound on the experiment count required to learn a local Lindbladian.
\begin{theorem}[Experiment count lower bound]\label{thm:lower_bound_measurement}
    For a fixed $k=O(1)$, consider learning an unknown $k$-local Lindbladian $\cL$ with access to the channel semigroup $\{e^{t\cL}\}_{t\geq 0}$. Given $\Lambda > 0$ and $\eps \in (0, \Lambda/16]$, if a learning algorithm can estimate coefficients of $k$-local Lindbladians with local dynamical strength at most $\Lambda$ up to $\ell^\infty$ error $\eps$ with success probability at least $2/3$, then the experiment count $N$ must satisfy
    \[
    N = \Omega\ab\Big(\frac{\Lambda^2}{\eps^2}).
    \]
\end{theorem}
\begin{proof}
    We consider Lindbladians that are either $\cL_{\gamma_0}$ or $\cL_{\gamma_1}$, as defined in \cref{sec:lower_bound_techniques}. As noted, any algorithm that learns the Lindbladian to $\ell^\infty$ error $\eps$ can distinguish $\gamma_0$ from $\gamma_1$. We will show that the two Lindbladians are hard to distinguish unless we perform sufficiently many experiments.

    Recall that the channel semigroup $e^{t\cL_{\gamma_i}}$ is a family of Pauli dephasing channels
    \[
    e^{t\cL_{\gamma_i}}(\rho) = (1-p_{\gamma_i t})\rho + p_{\gamma_i t} Z\rho Z, \qquad p_{\gamma_i t} \coloneq \frac{1- e^{-2\gamma_i t}}{2}.
    \]
    Fix a time $t \geq 0$ and define $p_i \coloneq p_{\gamma_i t}$ for $i \in \{0,1\}$ for the ease of notation. For any input state $\rho$, the output under $\cL_{\gamma_i}$ is
    \[
    \sigma_i \coloneq (1- p_i) \rho + p_i Z \rho Z.
    \]
    We introduce an auxiliary flag indicating whether $Z$ is applied or not,
    \[
    \tilde{\sigma}_i \coloneq (1- p_i) \ketbra{0}{0}\otimes \rho + p_i \ketbra{1}{1} \otimes Z\rho Z.
    \]
    Then by the data-processing inequality (\cref{fact:data_processing}, tracing out the flag), we have
    \[
    D(\sigma_0 \parallel \sigma_1) \leq D(\tilde\sigma_0 \parallel \tilde\sigma_1).
    \]
    The block-diagonal relative entropy gives
    \[
    D(\tilde\sigma_0 \parallel \tilde\sigma_1) = D_\rmb(p_0 \parallel p_1) + (1-p_0) D(\rho \parallel \rho) + p_0 D(Z \rho Z \parallel Z\rho Z) = D_\rmb(p_0 \parallel p_1),
    \]
    where $D_\rmb$ is the binary relative entropy defined as
    \[
    D_\rmb(p_0 \parallel p_1) \coloneq p_0 \ln \frac{p_0}{p_1} + (1-p_0) \ln\frac{1-p_0}{1-p_1}.
    \]
    Therefore we have
    \[
    D(\sigma_0 \parallel \sigma_1) \leq D_\rmb(p_0 \parallel p_1).
    \]
    Now bound $D_\rmb(p_0 \parallel p_1)$ uniformly over $t \geq 0$. For $t=0$, we have $p_0=p_1=0$ and $D_\rmb(p_0 \parallel p_1)=0$, so we may assume $t>0$, which ensures $p_0, p_1 \in (0,1)$. Let $\Delta \coloneq \gamma_1 - \gamma_0 = 4\eps$, $x \coloneq 2\gamma_0 t$, and $\beta \coloneq \Delta/\gamma_0$. Then we can write
    \[
    p_0 = \frac{1- e^{-x}}{2}, \qquad p_1 = \frac{1-e^{-(1+\beta)x}}{2}.
    \]
    By \cref{fact:chi2}, the $\chi^2$-divergence bound gives
    \begin{align*}
        D_\rmb(p_0 \parallel p_1) \leq  \frac{(p_0 - p_1)^2}{p_1(1-p_1)} &= \frac{e^{-2x}(1-e^{-\beta x})^2}{1-e^{-2(1+\beta)x}}\\
        &\leq \frac{e^{-2x}(1-e^{-\beta x})^2}{1-e^{-2x}}\\
        &\leq \frac{e^{-2x}\beta^2 x^2}{1-e^{-2x}} &\mbox{(by $1 - e^{-\beta x}\leq \beta x$)}\\
        &= \frac{\beta^2 x^2}{e^{2x}-1}\\
        &\leq \beta^2 = \frac{\Delta^2}{\gamma_0^2}.  &\mbox{(by $e^{2x}-1 \geq x^2$)}
    \end{align*}
    We now consider an adaptive algorithm performing $N$ experiments, which may apply arbitrary quantum operations, use ancillary qubits, choose future evolution times based on previous outcomes, and perform arbitrary final measurements. Let $\mathrm{dist}_0$ and $\mathrm{dist}_1$ be the final output distributions of the algorithm, when the Lindbladian is $\cL_{\gamma_0}$ and $\cL_{\gamma_1}$, respectively. We can show that the previous bound is stable under adaptivity. For any two input states $\rho_0$ and $\rho_1$, we have
    \[
      D\ab\big(e^{t\cL_{\gamma_0}}(\rho_0) \parallel e^{t\cL_{\gamma_1}}(\rho_1))
      \leq D(\rho_0 \parallel \rho_1) + D_\rmb(p_0 \parallel p_1) \leq D(\rho_0 \parallel \rho_1) + \frac{\Delta^2}{\gamma_0^2},
    \]
    which follows from the same flagged construction as above. This means that each experiment gives an additive increase of at most $\Delta^2/ \gamma_0^2$ in the relative entropy, and subsequent quantum operations can only decrease it by the data-processing inequality in \cref{fact:data_processing}. Starting from the same initial state, by induction over the $N$ experiments, the final output distributions satisfy
    \begin{equation}\label{eq:rel_entropy_bound}
        D(\mathrm{dist}_0 \parallel \mathrm{dist}_1) \leq N\frac{\Delta^2}{\gamma_0^2}.
    \end{equation}

    Assume the algorithm learns both $\cL_{\gamma_0}$ and $\cL_{\gamma_1}$ with success probability at least $2/3$. This implies that $\cL_{\gamma_0}$ and $\cL_{\gamma_1}$ can be distinguished. From the estimate $\hat\gamma$, define the binary test $\sfT \coloneq \mathbf{1}[\hat\gamma \geq \gamma_0 + 2\eps]$.
    \begin{itemize}
        \item In case of $\cL_{\gamma_0}$: success of the algorithm implies $\abs{\hat{\gamma} - \gamma_0} \leq \eps$, so $\hat{\gamma}\leq \gamma_0 + \eps < \gamma_0 + 2\eps$ and $\sfT = 0$ with probability at least $2/3$.
        \item In case of $\cL_{\gamma_1}$: success of the algorithm implies $\abs{\hat{\gamma} - \gamma_1} \leq \eps$, so $\hat{\gamma}\geq \gamma_1 - \eps > \gamma_0 + 2\eps$ and $\sfT = 1$ with probability at least $2/3$.
    \end{itemize}
    Therefore, the total variation distance between the two output distributions must satisfy
    \[
    \dtv(\mathrm{dist}_0, \mathrm{dist}_1) \geq \frac{1}{3}.
    \]
    Pinsker's inequality in \cref{fact:pinsker} gives
    \[
    \dtv(\mathrm{dist}_0, \mathrm{dist}_1) \leq \sqrt{\frac{1}{2} D(\mathrm{dist}_0 \parallel \mathrm{dist}_1)}.
    \]
    Combining with \cref{eq:rel_entropy_bound}, we have
    \[
    N \frac{\Delta^2}{\gamma_0^2} \geq \frac{2}{9}.
    \]
    Substituting $\gamma_0=\Lambda/4$ and $\Delta=4\eps$ yields the lower bound
    \[
    N \geq \frac{2}{9} \cdot \frac{\gamma_0^2}{\Delta^2} = \frac{2}{9} \cdot \frac{(\Lambda/4)^2}{16\eps^2}=\Omega\ab\Big(\frac{\Lambda^2}{\eps^2}). \qedhere
    \]
\end{proof}

\subsection{Total evolution time lower bound}
We then prove the lower bound on the total evolution time of the semigroup $e^{t\cL}$ required to learn a local Lindbladian $\cL$.
\begin{theorem}[Evolution-time lower bound]\label{thm:lower_bound_time}
    For a fixed $k=O(1)$, consider learning an unknown $k$-local Lindbladian $\cL$ with access to the channel semigroup  $\{e^{t\cL}\}_{t\geq 0}$. Given $\Lambda > 0$ and $\eps \in (0, \Lambda/16]$, if an algorithm can estimate coefficients of $k$-local Lindbladians with local dynamical strength at most $\Lambda$ up to $\ell^\infty$ error $\eps$ with success probability at least $2/3$, then the total evolution time $T_\mathrm{tot}$ must satisfy
    \[
    T_\mathrm{tot} = \Omega\ab\Big(\frac{\Lambda}{\eps^2}).
    \]
\end{theorem}
\begin{proof}
    We construct the same hard instances $\cL_{\gamma_0}$ and $\cL_{\gamma_1}$ as defined in \cref{sec:lower_bound_techniques}. Let $\Delta \coloneq \gamma_1 - \gamma_0 = 4\eps$, $x \coloneq 2\gamma_0 t$, and $\beta \coloneq \Delta/\gamma_0$. Then we can write
    \[
    p_0 = \frac{1- e^{-x}}{2}, \qquad p_1 = \frac{1-e^{-(1+\beta)x}}{2}.
    \]
    Recall that we have previously proved a uniform bound of $D_\rmb(p_0 \parallel p_1) \leq \Delta^2/\gamma_0^2$ over all $t \geq 0$ in the proof of \cref{thm:lower_bound_measurement}. 

    Now we slightly change the proof to obtain a bound that also depends on time $t$. Using the $\chi^2$-divergence bound in \cref{fact:chi2}, we have
    \begin{align*}
        D_\rmb(p_0 \parallel p_1) \leq  \frac{(p_0 - p_1)^2}{p_1(1-p_1)} &= \frac{e^{-2x}(1-e^{-\beta x})^2}{1-e^{-2(1+\beta)x}}\\
        &\leq \frac{e^{-2x}(1-e^{-\beta x})^2}{1-e^{-2x}}\\
        &\leq \frac{e^{-2x}\beta^2 x^2}{1-e^{-2x}} & \mbox{(by $1 - e^{-\beta x}\leq \beta x$)}\\
        &= \frac{\beta^2 x^2}{e^{2x}-1}\\
        &\leq \beta^2 \frac{x}{2}  & \mbox{(by $e^{2x}-1 \geq 2x$)}\\
        &= \frac{\Delta^2}{\gamma_0^2} \cdot \frac{2\gamma_0 t}{2}\\
        &= \frac{\Delta^2}{\gamma_0}t.
    \end{align*}
    The rest of the proof follows the same argument as the proof of \cref{thm:lower_bound_measurement}. The only difference is that we change the upper bound of $D_\rmb(p_0 \parallel p_1)$ from $\Delta^2/\gamma_0^2$ to $\Delta^2t/\gamma_0$. For any two input states $\rho_0$ and $\rho_1$, we have
    \[
      D\ab\big(e^{t\cL_{\gamma_0}}(\rho_0) \parallel e^{t\cL_{\gamma_1}}(\rho_1))
      \leq D(\rho_0 \parallel \rho_1) + D_\rmb(p_0 \parallel p_1) \leq D(\rho_0 \parallel \rho_1) + \frac{\Delta^2}{\gamma_0}t.
    \]
    This means that each experiment with evolution time $t$ gives an additive increase of at most $\Delta^2 t/ \gamma_0$ in the relative entropy.

    Now consider an algorithm that performs $N$ experiments, with evolution times $t_1, t_2, \dots, t_N$. Define the total evolution time
    $T_\mathrm{tot} \coloneq \sum_{j=1}^N t_j$.
    Let $\mathrm{dist}_0$ and $\mathrm{dist}_1$ be the final output distributions of the algorithm, when the Lindbladian is $\cL_{\gamma_0}$ and $\cL_{\gamma_1}$, respectively. Starting from the same initial state, it follows immediately by induction that
    \begin{equation}\label{eq:rel_entropy_bound_t}
        D(\mathrm{dist}_0 \parallel \mathrm{dist}_1) \leq \sum_{j=1}^N \frac{\Delta^2}{\gamma_0}t_j = \frac{\Delta^2}{\gamma_0} T_\mathrm{tot}.
    \end{equation}

    Assume the algorithm learns both $\cL_{\gamma_0}$ and $\cL_{\gamma_1}$ with success probability at least $2/3$. This implies that $\cL_{\gamma_0}$ and $\cL_{\gamma_1}$ can be distinguished. Therefore, the total variation distance between the two output distributions must satisfy
    \[
    \dtv(\mathrm{dist}_0, \mathrm{dist}_1) \geq \frac{1}{3}.
    \]
    Pinsker's inequality in \cref{fact:pinsker} gives
    \[
    \dtv(\mathrm{dist}_0, \mathrm{dist}_1) \leq \sqrt{\frac{1}{2} D(\mathrm{dist}_0 \parallel \mathrm{dist}_1)}.
    \]
    Combining with \cref{eq:rel_entropy_bound_t}, we have
    \[
    \frac{\Delta^2}{\gamma_0} T_\mathrm{tot} \geq \frac{2}{9}.
    \]
    Substituting $\gamma_0=\Lambda/4$ and $\Delta=4\eps$ yields the lower bound
    \[
    T_\mathrm{tot} \geq \frac{2}{9} \cdot \frac{\gamma_0}{\Delta^2} = \frac{2}{9} \cdot \frac{\Lambda/4}{16\eps^2}=\Omega\ab\Big(\frac{\Lambda}{\eps^2}). \qedhere
    \]
\end{proof}

The argument above bounds the information gained from a realized sequence of evolution times.  In a fully adaptive protocol, however, the learner may decide both the evolution time of the next experiment and when to stop based on the outcomes observed so far, so the total evolution time is itself a random variable. The same per-unit-time information bound still applies after conditioning on the past, which leads to the following expected-time version.

\begin{corollary}[Evolution time lower bound for adaptive learning]\label{cor:expected_time_lower}
Under the assumptions of \cref{thm:lower_bound_time}, allow the algorithm to choose a stopping time $N$ and evolution times $t_1,t_2,\ldots,t_N$ adaptively from previous outcomes and internal randomness. Let
$T_{\mathrm{tot}}\coloneq \sum_{j=1}^N t_j$ be the resulting random total evolution time. Given $\Lambda > 0$ and $\eps \in (0, \Lambda/16]$, if an algorithm succeeds with probability at least $2/3$ over all $k$-local Lindbladians $\cL$ with local dynamical strength at most
$\Lambda$, estimating their coefficients up to $\ell^\infty$ error $\eps$,
then its worst-case expected total evolution time satisfies
\[
  \sup_{\cL}
  \E_{\cL}\ab[T_{\mathrm{tot}}]
  =
  \Omega\ab\Big(\frac{\Lambda}{\eps^2}).
\]
\end{corollary}
\begin{proof}
We reuse the hard instances $\cL_{\gamma_0},\cL_{\gamma_1}$ and write $\mathbb{P}_{\gamma_i}^{\mathrm{tr}}$ for the law of the fully adaptive transcript when the true generator is $\cL_{\gamma_i}$, $i\in\{0,1\}$.  The transcript includes the algorithm's internal randomness, the chosen evolution times, all measurement outcomes, and the final output.  Let $\E_{\gamma_i}$ denote expectation under $\mathbb{P}_{\gamma_i}^{\mathrm{tr}}$.

We first prove the relative-entropy bound for a truncated transcript.  Fix $N_0$ and run the algorithm until it stops or until $N_0$ experiments have been performed. After stopping, set all later evolution times to zero and append dummy outcomes.  Let $\mathbb{P}_{\gamma_i}^{(N_0)}$ be the law of this truncated transcript.  Conditioned on the same realized history before the $j$-th possible experiment, the adaptive rule fixes the next evolution time $t_j$ and the admissible operation preparing the input.  The actual system-memory input states under the two hypotheses may be different, but the bound for a single experiment in the proof of \cref{thm:lower_bound_time} applies to arbitrary pairs of input states after tensoring the dephasing channel with the identity on any memory and ancillas.  Hence, the $j$-th experiment can increase the relative entropy by at most $\frac{\Delta^2}{\gamma_0}t_j$.  The intervening adaptive operations, measurements, and classical post-processing cannot increase relative entropy by data processing.  Therefore, the chain rule for relative entropy gives
\[
  D\ab\Big(\mathbb{P}_{\gamma_0}^{(N_0)}\parallel \mathbb{P}_{\gamma_1}^{(N_0)}) \leq
  \frac{\Delta^2}{\gamma_0}\, \E_{\gamma_0}\ab\bigg[\sum_{j=1}^{N_0}t_j].
\]
Taking $N_0 \to\infty$ and using monotone convergence for the right-hand side, together with the monotonicity of relative entropy for increasing transcripts, yields
\[
  D\ab\big(\mathbb{P}_{\gamma_0}^{\mathrm{tr}}\parallel \mathbb{P}_{\gamma_1}^{\mathrm{tr}}) \leq \frac{\Delta^2}{\gamma_0}\,\E_{\gamma_0}\ab[T_{\mathrm{tot}}].
\]

Let $\mathrm{dist}_i$ be the law of the final output under $\cL_{\gamma_i}$.  Since the final output is a function of the transcript, data processing gives
\[
  D(\mathrm{dist}_0\parallel \mathrm{dist}_1) \leq D\ab\big(\mathbb{P}_{\gamma_0}^{\mathrm{tr}}\parallel \mathbb{P}_{\gamma_1}^{\mathrm{tr}}) \leq
  \frac{\Delta^2}{\gamma_0}\,\E_{\gamma_0}\ab[T_{\mathrm{tot}}].
\]
As in \cref{thm:lower_bound_time}, success probability at least $2/3$ on both hard instances yields a test distinguishing $\cL_{\gamma_0}$ from $\cL_{\gamma_1}$ with total variation distance at least $1/3$. Pinsker's inequality then gives
\[
  D(\mathrm{dist}_0\parallel \mathrm{dist}_1)\geq\frac{2}{9}.
\]
Combining the last two displays and substituting $\gamma_0=\Lambda/4$ and $\Delta=4\eps$ gives
\[
  \E_{\gamma_0}\ab[T_{\mathrm{tot}}] \geq \frac{2}{9}\cdot \frac{\gamma_0}{\Delta^2}
  = \Omega\ab\Big(\frac{\Lambda}{\eps^2}).
\]
Since $\cL_{\gamma_0}$ satisfies the local-strength constraint, the same lower bound holds for the worst-case expected total evolution time.
\end{proof}

\section*{Acknowledgments}
    The authors thank Jinge Bao for discussions at an early stage of this work and Wenjun Yu for helpful comments on the manuscript.
    This project is supported by the National Research Foundation, Singapore through the National Quantum Office, hosted in A*STAR, under its Centre for Quantum Technologies Funding Initiative (S24Q2D0009) and its Advanced Quantum Algorithms and Solutions Funding Initiative (S25Q9DA001 and S25Q9DA002).
    N.G.\ acknowledges support by the SandboxAQ Research Excellence Scholarship.
    Z.C., N.G., and Z.Y.\ acknowledge support by the CQT Young Researcher Career Development Grant.
\bibliographystyle{alphaurl}
\bibliography{ref}

\end{document}